\documentclass[12pt]{article}
\usepackage{natbib}
\usepackage{amsmath}
\usepackage{amsthm}
\usepackage{amssymb}
\usepackage{algorithm}
\usepackage{algpseudocode}
\usepackage{graphicx}
\usepackage{sidecap}
\usepackage{subcaption}
\usepackage{caption}
\usepackage{comment}
\usepackage{float}
\usepackage{lipsum}
\usepackage{changepage}
\usepackage{tikz}
\usepackage{pgf, pgfplots}
\usepackage{tikz-network}
\pgfplotsset{compat=1.17} 
\usepackage{xcolor}
\usetikzlibrary{arrows}
\usetikzlibrary{calc}
\usetikzlibrary{positioning}
\usetikzlibrary{automata}
\usetikzlibrary{shapes}
\usetikzlibrary{shapes.geometric}
\usetikzlibrary{decorations.markings}
\usepackage{setspace}
\usepackage{fullpage}
\usepackage{mathptmx}
\usepackage{pdfpages}
\usepackage{import}
\usepackage{rotating}
\usepackage{hyperref}
\usepackage{cleveref} 
\usepackage{minitoc}
\usepackage{algpseudocode} 
\usepackage{todonotes}
\usepackage{xfrac}

\newtheorem{thm}{Theorem}       
\newtheorem*{thm*}{Theorem}
\newtheorem{proposition}{Proposition}
\newtheorem*{prop*}{Proposition} 
\newtheorem{cor}{Corollary}     
\newtheorem{lemma}{Lemma}
\newtheorem*{lem*}{Lemma}
\newtheorem{definition}{Definition}
\newtheorem*{defn*}{Definition}
\newtheorem{rem}{Remark} 
\newtheorem*{rem*}{Remark} 
 
\newtheorem*{ex*}{Example} 

\newtheorem*{ass*}{Assumption}

\usepackage{comment}
\definecolor{green}{HTML}{51A351}
\definecolor{blue}{HTML}{2F96b4}
\definecolor{gray}{HTML}{BCBCBC}
\definecolor{gray1}{rgb}{0.62,0.65,0.65} 
\definecolor{gray2}{rgb}{0.29,0.31,0.31} 
\definecolor{gray3}{rgb}{0.06,0.07,0.08} 
\definecolor{teal1}{rgb}{0.13,0.54,0.52} 
\definecolor{teal2}{rgb}{0.18,0.34,0.32} 
\definecolor{brown1}{rgb}{0.59,0.53,0.49} 
\definecolor{brown2}{rgb}{0.81,0.73,0.64} 
\definecolor{yellow1}{RGB}{228,182,19} 
\definecolor{blue_low}{RGB}{137,187,255} 
\definecolor{blue_high}{RGB}{74,141,255} 
\definecolor{red_low}{RGB}{255,123,123} 
\definecolor{red_high}{RGB}{255,0,0} 

\newcommand{\monthyeardate}{\ifcase \month \or January\or February\or March\or April\or May \or June\or July\or August\or September\or October\or November\or December\fi, \number \year}
\title{Persuasion in Networks: Can the Sender Do Better than Using Public Signals?}

\author{Yifan Zhang\footnote{ Zhang: Christ's College, University of Cambridge (email: \texttt{yz558@cam.ac.uk}). I am especially grateful to Matt Elliott for guidance and comments throughout the project. I'm also grateful to Alastair Langtry and participants at the Cambridge networks workshop, Warwick PhD conference and GAMES 2024 for helpful comments and suggestions. This work receives support from the Faculty of Economics and Keynes Fund. Any remaining errors are the sole responsibility of the author. }}

\date{\today \\ \vspace{7mm}} 
\begin{document}
\maketitle
\begin{abstract}

Political and advertising campaigns increasingly exploit social networks to spread information and persuade people. This paper studies a persuasion model to examine whether such a strategy is better than simply sending public signals. Receivers in the model have heterogeneous priors and will pass on a signal if they are persuaded by it to take sender's preferred action. I show that a risk-neutral or risk-loving sender prefers to use public signals, unless the more sceptical receivers are sufficiently more connected in the network. This result still holds when the network exhibits homophily. When the sender is risk averse or has a threshold function as payoff function, I characterise conditions under which the sender prefers to exploit the network. An important factor identified by these conditions is whether the less sceptical receivers are sufficiently connected.

\noindent \\

\noindent  
\emph{Keywords:} Bayesian persuasion, networks, random networks, diffusion \\ 
\end{abstract}

\maketitle
\newpage

\onehalfspacing
\normalsize
	People are increasingly relying on social networks to gather information. Many argue that information going viral on social media had significant influences on people's decisions related to elections and public health\footnote{See e.g. \citet{2005Are}, \citet{2017Social} and references therein }. Realising the power of such viral information, firms and politicians have spent a vast amount of resources on social media campaigns. Their spending on social media is also increasing relative to traditional mass media\footnote{See, for example, \citet{Nelson2021} for political advertisements and the 28th edition of CMO survey for commercial advertisements( https://cmosurvey.org/results/28th-edition-february-2022/ ) }. A natural question is how effective such social media campaigns are in influencing people's decisions. In particular, are they more effective than mass media?

    To see the trade-off between using mass media and social networks, consider an example in which a firm tries to advertise its product. Some consumers are sceptical about the quality of the product, while others have prior belief that it has high quality. In order to persuade the sceptics to buy, the firm can conduct quality tests or get pre-release reviews, and send sceptics this information. Ideally, the firm would like to send this information to the sceptics but keep them away from the believers, since negative results from the tests and negative reviews could dissuade the believers. If the information is broadcast publicly, everyone must receive the same information. If the firm spreads information on social networks, some people may not receive information due to structure of the network and the way people share information, and the firm may be able to exploit this. In particular, the firm could try to target information towards the sceptics. However, communication over the network endogenously imposes restrictions on correlations between what believers and sceptics observe, so it is unclear whether and when the firm can succeed in sending information more often to sceptics than to believers. 
    
    In this setup, does the firm want to publicly broadcast information or spread information on the network? What are the networks in which public broadcasts can and cannot be outperformed? How do the firm's preferences influence the comparison?

To address these questions, I propose the following persuasion model. There are two states and receivers have two actions $\{0,1\}$. 
There are infinitely many receivers, who want to match their action with the state. The sender wants receivers to take action 1 regardless of the state. Receivers have heterogeneous priors. If receivers make decisions based on priors, some of them (believers) would choose action 1 while others (sceptics) would choose action 0. Receivers are located on a network captured by a random network model. Receivers know their own degree. 

 The sender can use public signals or network signals. With public signals, all receivers must receive the same realised signal. If the sender uses network signals, then the sender can select a small proportion of receivers as seeds based on their priors and positions in the network. These seeds will receive signals from the sender. When receivers observe some signal, they decide whether to pass it on to their neighbours. In the main analysis, I assume that receivers pass on a signal if and only if they are persuaded by it to take the sender's preferred action. Intuitively, if people are persuaded to purchase a product or vote for a politician, they are willing to share its advertisement to their friends.

 I allow the sender to send no signal. This plays no role with public signals. If the sender uses network signals, the receivers may not receive any signals due to the network structure. When receivers observe no signal, they cannot distinguish whether it is due to network structure or sender choosing to send no signal. The sender could potentially exploit this in her benefit. 

The main research question is whether and when network signals can outperform public signals for the sender. The answer to this question in turn depends on whether the sender can use network signals and seeding to send information more often to sceptics than to believers. Also, some sceptics may not receive any signals if the sender uses network signals. This may be beneficial or costly to the sender, depending on the sender's preference.

The first result states that in the baseline network (i.e. one without homophily or correlations between degree and priors), if the sender is risk loving or risk neutral, network signals cannot outperform public signals. The main intuition underlying this result is that the sender needs signals to get viral in order to reach a large population, but it is difficult to control who are reached by a viral signal using seeding. Also, believers and sceptics will be ex ante symmetric in terms of their positions in such baseline networks. 
Therefore, there are no differences between how sceptics and believers are exposed to signals, and the sender cannot send information more often to sceptics than to believers.

One might suspect that homophily allows the sender to exploit seeding and provide information more often to sceptics than to believers. However, introducing homophily and/or making believers more connected than sceptics does not change the above result. Networks exhibiting these features do introduce asymmetry between believers and sceptics. However, the asymmetry makes believers more exposed to signals than sceptics, which is the opposite of what the sender desires. 

If sceptics are sufficiently more connected in the network than believers, the sender prefers network signals. Sceptics being more connected allows the sender to send information to them more often than to the believers. 

Next, I study the case where the sender is risk averse. When believers are sufficiently connected in the network, 
the sender can always persuade some sceptics using network signals, but it is in general impossible to persuade all sceptics at the same time. The reverse is true for public signals. Therefore, when believers are sufficiently connected, a sufficiently risk-averse sender prefers network signals. 

Finally, I study a voting game, in which sender's payoff function is a threshold function. Like in the case of a risk-averse sender, it is found that whether believers are sufficiently connected is an important factor influencing whether the sender prefers network signals. 

In summary, this paper characterises conditions on the network structure and the sender's preference under which network signals can and cannot outperform public signals. It shows that in many cases, public signals perform better even if receivers have heterogeneous priors. The conditions also illustrate that the sender's preference will influence which aspect of the network matters for the comparison between public signals and network signals. 


In practice, the sender in many situations has access to tools for publicly broadcasting information. The results imply that, in many cases, resources may be more efficiently spent on sending public signals, rather than on trying to understand network structures and picking the optimal seeds. In addition, persuasion on large networks is complicated while public signals have been well studied. The sender's value of using network signals can be very difficult to compute. From these comparisons, we can use sender's value of using public signals to obtain what sender could get at most or at least using network signals.

Relation and contribution to the literature are discussed in the next section. \Cref{sec:model} presents the model. \Cref{sec:results} presents the main results, and \Cref{sec:analysis} contains lemmas useful for understanding the main results. \Cref{sec:extension} discusses implications of relaxing some assumptions and possible extensions. The final section concludes.

\section{Literature Review} \label{sec:literature}
This paper is most related to the literature on persuasion with information spillovers. \citet{galperti2023spillover} studies an environment where a sender can directly send signals to some receivers (seeds), and receivers observe the same information as those connected to them (by a directed path). They show that the problem is equivalent to allowing the sender to directly send information to all receivers, by transforming the network appropriately. They provide a characterization of feasible outcomes and show how it changes with the network and seeds. \citet{babichenko2021multi} studies a similar problem, and also explores the algorithmic aspect of the problem. \citet{kerman2021persuading} studies a voting game in which receivers observe their neighbours' signals. They characterise sender's optimal signals in some stylised networks, and compare sender's payoff to the case when receivers do not communicate. \citet{liporace2021persuasion} assumes that receivers observe signals of their neighbours, and receivers have heterogeneous priors. She characterises the sender's optimal strategy, restricting attention to strategies that persuade everyone in the good state. 
In \citet{egorov2020persuasion}, information is shared between receivers with some probability, and receivers choose whether to subscribe to the sender which is costly. They characterise the sender's optimal strategy in some stylised networks. 

This strand of literature usually assumes that information will be transmitted mechanically to one's neighbours, while this paper studies more endogenous spread of information. Another key difference is that, the sender in most of these papers can always make sure that all receivers observe information from the sender. They do not consider the possibility that sender's information may not reach some receivers, which is very likely when considering large social networks. Also, in this paper, the sender and receivers do not have full information of the network, whereas the literature studies fixed and known networks. Uncertainties over network is a plausible assumption in large networks. 

Also related is the literature on strategic communication, which studies how players communicate by cheap talk when their payoffs depend on each other's actions. Most papers in this literature (e.g. \citet{galeotti2013strategic} and \citet{hagenbach2010strategic}) focus on conditions under which players truthfully communicate with each other, and do not study persuasion. The exception is \citet{squintani2024persuasion}. \citet{squintani2024persuasion} studies an environment where there is one sender and one receiver (randomly selected from all players). He characterises conditions for truthful communication and studies the optimal network and network formation. 

Some papers have looked at the comparison between public signals and private signals, but do not allow communication between receivers. \citet{wang2013bayesian} compares public signals and i.i.d. signals in a voting game, and finds that public signal is always better. \citet{inostroza2018persuasion} also compares public and private signals, imposing certain restrictions on private signals. \citet{candogan2020} studies a problem of a social media platform to signal the quality of content, and finds that public signals are sufficient if the goal is to minimize misinformation. \citet{arieli2019private} relates optimality of public signals to the core of a cooperative game. 
Communication between receivers puts endogenous restrictions on joint distributions over what receivers observe when sender sends personalised signals. Due to this and other more detailed differences in setup, the trade-offs and insights in this paper are different from existing papers in this literature. 

Finally, this paper uses a random network model. Random networks have been used in \citet{akbarpour2020just}, \citet{sadler2020} and \citet{campbell2013word} to model social networks among others. In particular, it is very often used to model diffusion of ideas and information, see e.g. \citet{golub2012homophily}, \citet{sadler2023influence} and \citet{merlino2023debunking}. The problems studied in these papers are very different from this paper. 
\section{Model}\label{sec:model}

\subsection{Agents and Network}
I consider an environment with one sender, and an infinite and countable set of receivers. There are two possible states of the world, $\omega \in \{0,1\}$, and each receiver has two possible actions $\{0,1\}$. The receivers want to match their action with the state. 
The receiver gets a payoff of 1 if their action matches the state, and gets 0 otherwise. The sender wants receivers to choose action 1 regardless of the state. If a proportion $x$ of the receivers choose action 1, then the sender gets payoff $v(x)$, where the function $v$ is continuous, weakly increasing and bounded. I normalise $v(0)=0$.

The receivers' payoff function implies that action 1 is optimal for receivers if and only if they believe the probability of state 1 is greater than 0.5. Throughout the paper, I break indifference in favour of the sender\footnote{Results can be extended if I break indifference in the other way, and proofs are essentially the same.}, i.e. I assume that receivers choose action 1 when they believe the probability of state 1 is exactly 0.5.

A proportion $\gamma_{h}$ of receivers are of type $h$ and have prior belief that state $\omega$ realises with probability $\mu_{h}(\omega) $. The rest of the receivers are of type $l$, and have a prior belief $\mu_{l}(\omega)$ for state $\omega$. I assume that $\mu_{h}(\omega=1) > 0.5 >\mu_{l}(\omega=1)$. I denote the proportion of type $l$ receivers by $\gamma_{l}$. Type $h$ receivers are often referred to as \textit{believers}, and type $l$ receivers as \textit{sceptics}. 

The sender has a prior belief that state $\omega$ realises with probability $\mu_{s}(\omega)$. The sender can send signals from the signal space $S \cup \{\emptyset \}$, where $S$ is finite and contains at least two elements\footnote{All results below easily extend to cases where $S$ has only one element or $S$ is empty. If $S$ is empty, the result is trivial since the sender essentially cannot do anything. If $S$ contains one element, payoff from network signals weakly decreases, while it can be easily shown that payoff using public signals does not change. For results where network signals perform better than public signals, in their proofs, sender only uses 1 signal from $S$ with positive probability.}. The element $\emptyset$ captures the sender's ability to send no signal. 
The sender designs a signal structure $\pi: \{0,1\} \rightarrow \Delta(S \cup \{\emptyset\})$. The probability of signal $s \in S \cup \{\emptyset \}$ being realised in state $\omega$ is $\pi(s|\omega)$, and $\pi$ must satisfy $\sum_{s \in S \cup \{\emptyset \}} \pi(s|\omega)=1$ for $\omega=0,1$. 
 If the sender uses network signals, the sender can also choose where to seed the signals. Precisely how seeding works and timing of the game are described in the next two subsections. 

Now we describe the random network model. Each receiver $i$ has \textit{ connectedness} $\lambda_{i}$. With $n$ receivers, if $i$ and $j$ have the same prior, the probability that they are linked is $min \{ \frac{\lambda_{i}\lambda_{j}}{n},1 \}$. If they have different priors, the probability that they are linked is $min \{ \frac{q \lambda_{i}\lambda_{j}}{n},1 \}$ where $q \in (0,1]$. Intuitively, $q$ controls homophily. Whether $i$ and $j$ are linked is independent of links between all other pairs (conditional on connectedness of $i$ and $j$). Conditional on having type $t$, a receiver has connectedness $\lambda$ with probability $f_{t}(\lambda)$, where $f_{t}(\cdot)$ is a probability mass function with finite support\footnote{The analysis can be extended to the case where distribution over $\lambda$ is continuous and  unbounded. This will make restrictions on these distributions and the proofs more tedious, without adding much insight. \citet{bollobas2007phase} states the precise restrictions on these distributions when they have continuous and unbounded support.}.

Note that the probability that two nodes are linked depends on the number of receivers $n$. There are infinitely many receivers in the model, so we look at the limiting network when $n$ goes to infinity, as in many other papers using random network models. 
The network formation process is common knowledge. Each receiver $i$ in addition knows her own degree, denoted by $d_{i}$ (degree is the number of neighbours a receiver has). When studying network signals, I use receivers and nodes interchangeably. 

The random network model is borrowed from \citet{bollobas2007phase}. As discussed in the literature review, random network models are widely applied to model large social networks. In particular, they have been used to capture networks over which people communicate information and opinions. 

Two important network statistics are degree distribution and the probability that a node's neighbour has the same type as the node. Theorem 3.13 in \citet{bollobas2007phase} shows that
\begin{rem} \label{rem:degree_dist}
    As $n \rightarrow \infty$, degree distribution of type $t$ receivers converges to $\{p_{d}^{t}\}_{d}$. Also, conditional on degree $d$ and type $t$ of a node, each of its neighbours has an independent probability $q_{t}$ of having the same type\footnote{More precisely, as $n \rightarrow \infty$, conditional on degree $d$ and type $t$ of a node, the probability that $d_{t}$ of its neighbours are of type $t$ converges to $Binom(d_{t};d,q_{t})$}, i.e. type $t$. The degree distribution $\{p_{d}^{t}\}_{d}$ and probability $q_{t}$ are 
    \[
    p_{d}^{t} =\sum_{\lambda'}f_{t}(\lambda') Poisson(d; \lambda' (E_{t}(\lambda) \gamma_{t}  + E_{t'}(\lambda)\gamma_{t'}q ) ) \qquad q_{t} = \frac{E_{t}(\lambda) \gamma_{t}  }{ E_{t}(\lambda) \gamma_{t}  + E_{t'}(\lambda)\gamma_{t'}q }
    \]
    where $E_{t}(\lambda) = \sum_{\lambda} f_{t}(\lambda)\lambda$ and $t' \ne t$
\end{rem}

\subsection{Network signals}
As described above, social network with infinitely many receivers is modelled as the limit of a sequence of finite random networks. Therefore, I will first describe the game with $n$ receivers. The sender's optimal payoff facing infinitely many receivers is captured as the limit of her optimal payoffs in these finite games when $n \rightarrow \infty$.

With $n$ receivers, the timing goes as follows. 
The sender moves first and chooses a signal structure $\pi$, and a seeding strategy $Z$. Given a network with $n$ receivers $g^{n}$ and a signal $s \in S$, 
$Z(g^{n},s)$ returns a subset of receivers. The subset selected by $Z$ must have a size smaller than $ n^{\alpha}$, where $\alpha<1$. Intuitively, $Z$ selects some seeds that will receive signals first, and seed selection can be conditioned on the network and the signal realised. 

After seeing the sender's strategy, the receivers choose their strategies, $a_{i}: (S \cup \emptyset) \times \mathbb{N}_{0} \rightarrow \{0,1\} $, where $a_{i}(s,d)$ describes what $i$ would do when she has degree $d$ and observes $s \in S \cup \{\emptyset\}$.

Then nature draws a network $g^{n}$ according to the network formation process described above. The nature draws a state, a signal $s$ according to $\pi$. If the signal is $\emptyset$, all receivers observe $\emptyset$. Intuitively, if the sender sends no signal, then all receivers observe no signal. 

If $s \ne \emptyset$, then the receivers $Z(g^{n},s)$ are selected as seeds. Given network $g^{n}$ and the seeds, a receiver $i$ observes $s$ if and only if at least one of the following conditions hold:
\begin{enumerate}
	\item $i$ is a seed. 
	\item $i$ has a neighbour $j$ who observes $s$, and $j$ chooses action 1 upon observing $s$, \\i.e. $a_{j}(s,d_{j})=1$.
\end{enumerate}
All the other receivers observe $\emptyset$.  In words, the condition here says that a receiver will pass on a signal to her neighbours, if and only if she is persuaded by the signal to take the sender's preferred action. 
Intuitively, people will pass on the advertisement of a product or a politician if they are convinced to buy the product or vote for the politician.

As in standard Bayesian persuasion, since receivers move after the sender, their strategies must constitute a (Nash) equilibrium in all subgames after the sender moves. Given receivers' equilibrium behaviour in the subgames, the sender chooses a strategy to maximise her expected payoff. Equilibrium in subgames and the sender's problem are described in more detail below. 

Let $\Pi \equiv (\pi,Z)$ denote the sender's network signal strategy and $\mathbf{a}_{-i}$ denote the strategy profile of receivers excluding $i$. Given $\Pi$ and $\mathbf{a}_{-i}$, receiver $i$ with type $t_{i}$ form a posterior upon observing each $s \in S \cup \{\emptyset\}$ and her degree $d \in \mathbb{N}_{0}$. The posterior belief for state 1 is denoted by $\psi_{net}(s,d;\mathbf{a}_{-i},\Pi,t_{i})$. Recall that receivers optimally choose action 1 if and only if they believe state 1 happens with probability higher than 0.5, so equilibrium in the subgame following $\Pi$ is defined as
\begin{definition}
   In a network signals game, receiver $i$'s strategy $a_{i}$ is a best response to $\Pi$ and $\mathbf{a}_{-i}$ if $a_{i}(s,d)=1 \iff \psi_{net}(s,d;\mathbf{a}_{-i},\Pi,t_{i}) \ge 0.5$ for all $s \in S \cup \{\emptyset\}$ and $d \in \mathbb{N}_{0}$. In the subgame following $\Pi$, receivers' strategy profile $\mathbf{a}$ is an equilibrium if $a_{i}$ is best responding to $\Pi$ and $\mathbf{a}_{-i}$ for all $i$.
\end{definition}
The next remark shows that, given any $\Pi$, the subgame following it has a unique equilibrium outcome.\footnote{This remark relies on breaking indifferences in favour of the sender as assumed above. Again, as stated above, results below still go through if we break indifferences against the sender.}.
\begin{rem} \label{rem:existence_of_eq}
    For any sender's strategy $\Pi$, all equilibria in the subgame following $\Pi$ induce the same distribution over receiver's action profile in each state. 
\end{rem}
The remark implies that the sender's expected payoff is the same in all equilibria in the subgame following $\Pi$. Let $V_{net}(\Pi,n)$ denote sender's expected payoff in the network signals game with $n$ receivers, when the sender chooses $\Pi$ and receivers play according to equilibrium in the subgame following $\Pi$. The sender's problem is to choose $\Pi$ and maximise $V_{net}(\Pi,n)$

For the game with $n$ receivers, we can find the sender's supremum payoff. Denote this payoff by $V^{*}_{net}(n)$. Sender's payoff facing infinitely many receivers is captured by $lim_{n \rightarrow \infty} V^{*}_{net}(n)$. However, this limit may not exist for some parameters\footnote{It can be shown that this only happens for a non-generic set of parameters, but the proof is very long and does not bring additional insight, so it is relegated to online appendix for the sake of space.}. Therefore, I compare $limsup_{n \rightarrow \infty} V^{*}_{net}(n)$ and $liminf_{n \rightarrow \infty} V^{*}_{net}(n)$ with the sender's payoff from public signals. 
\subsection{Public signals}
Again, I describe the game with $n$ receivers. The sender's optimal payoff facing infinitely many receivers is captured as the limit of her optimal payoffs in these finite games when $n \rightarrow \infty$\footnote{This is equivalent as just letting the sender face a continuum of receivers, of which $\gamma_{h}$ has type $h$. I describe sender's problem as the sequence of finite games to be consistent with the setup above for network signals}. 

First, the sender designs and commits to signal structure $\pi$. After seeing sender's strategy, receivers choose their strategies $a_{i}: S \cup \emptyset \rightarrow \{0,1\}$, where $a_{i}(s)$ describes receiver $i$'s action upon observing $s \in S \cup \emptyset$. 
The nature draws a state, and then a signal is realised according to $\pi$. All receivers observe the realised signal.

In the subgame following $\pi$, receivers best respond to $\pi$ and to each other. With public signals, all receivers directly receive the same signal from the sender, so the network and $\mathbf{a}_{-i}$ are irrelevant for forming posterior. Let $\psi_{pub}(s;\pi,t_{i})$ denote posterior of receiver $i$ with type $t_{i}$.
Equilibrium in subgames following sender's strategies are similarly defined
\begin{definition}
  In a public signals game, receiver $i$'s strategy $a_{i}$ is a best response to $\pi$ and $\mathbf{a}_{-i}$ if $a_{i}(s)=1 \iff \psi_{pub}(s;\Pi,t_{i}) \ge 0.5$ for all $s \in S \cup \{\emptyset\}$. In the subgame following $\pi$, receivers' strategy profile $\mathbf{a}$ is an equilibrium if $a_{i}$ is best responding to $\pi$ and $\mathbf{a}_{-i}$ for all $i$.
\end{definition}
A result analogous to remark \ref{rem:existence_of_eq} holds for public signals. Given receivers' equilibrium behaviour in subgames, the sender chooses $\pi$ to maximise expected payoff. We can find the sender's optimal payoff for the game with $n$ receivers, denoted by $V^{*}_{pub}(n)$. The sender's optimal payoff facing infinitely many receivers is captured by $V^{*}_{pub} \equiv lim_{n \rightarrow \infty} V^{*}_{pub}(n)$.

The main research question of this paper is to compare $V^{*}_{pub}$ and limiting payoff using network signals defined in the previous subsection.



It is worth re-emphasising that receivers' actions are conditioned on the signal (which can be empty) they receive in the public signals game. Similarly, their actions are conditioned on the signal they receive and their degree in the network signals game. As can be seen in the definitions of equilibrium in subgames, receivers optimise their action given their signal and degree if relevant. Therefore, although receivers choose strategies before realisation of signal and network, it is equivalent to letting them move after observing their signal and degree. This arrangement of timing is standard in Bayesian persuasion and when analysing Bayesian games in general. 
\subsection{Discussion of Setup}
\textbf{Rule of sharing signals:} The justification for the rule of sharing signals is that, after people purchase a certain product, they are more likely to talk about the product and thus pass information about it to their friends. \citet{banerjee2013diffusion} studies diffusion of micro-finance and finds evidence that people who use micro-finance are more likely to transmit information about it to their neighbours, which provides indirect support for the rule used in this paper. Sharing behaviour used here is similar to those used in \citet{campbell2013word} and \citet{sadler2020}. 

Also, main results are robust to other rules. In particular, another plausible rule of sharing is to share when the receivers' action agrees with the meaning carried by the signal. This is discussed in more detail in section \ref{sec:extension}.

\textbf{Seeds:} I assume that seeds are a vanishing fraction of the population, mainly for tractability. This is nonetheless plausible on large social networks. Targeted advertising services on social media usually use auctions to select which advertisement to present to a user, and the auction takes into account the relevance of the advertisement to the user. Intending to target a wide audience using these services will result in losing more often to other competing ads, and it is difficult to actually reach a wide audience directly. Results can still go through if seeds as a fraction of population is non-vanishing but small enough. I discuss this and how results could change if the seeds is a large fraction of the population in section \ref{sec:extension}.

I am also allowing the sender to use information about the whole network when selecting seeds, and in practice they can use much less information. Results on network signals being worse than public signals are stronger if we restrict the information used in seeding. Results on network signals being better do not rely on sender's ability to use so much information in seeding. 

\textbf{Empty Signal:} From description of the public signals game, the empty signal's role is not different from non-empty ones, so it is actually WLOG to assume that sender does not use empty signals in public signals game. 
For the network signals game, as we will see below, even if sender cannot use $\emptyset$, she can still essentially send no signal to almost everyone. Intuitively, this is done by choosing the seeds to be very isolated nodes. Allowing the sender to send $\emptyset$ directly is mainly for exposition. This also implies that one can think of $\emptyset$ as the sender choosing to conceal some signals. This interpretation requires sender to commit to concealment of signals before signal realisation. Relaxing the ability to commit to concealment plan restricts what sender can do using network signals, and thus strengthens results on public signals being better. For the results on network signals being better, in their proofs, the sender uses strategies where non-empty signals give the sender a weakly higher payoff than $\emptyset$, so the sender has no incentive to conceal them. 

\section{Analysis of the Network and Receivers' Posterior}\label{sec:analysis}

In order to compare sender's payoff from public and network signals, we first need to understand equilibrium in subgames given sender's strategies. This is straightforward for public signals, which is essentially the same as standard Bayesian persuasion. For network signals, analysis is more complicated since signals spread endogenously on networks. I present some lemmas that can be used to characterise limiting probabilities with which receivers observe each signal and $\emptyset$ when $n$ gets large. These lemmas can then be used to approximate and bound sender's payoffs using network signals in large networks. 

The next lemma pins down the receivers' posterior upon observing some non-empty signal $ s \ne \emptyset$. This lemma says if a receiver observes some non-empty signal $s$, then her posterior is the same as if all signals are sent publicly. 
\begin{lemma} \label{lemma:non_empty_posterior}
	In the network signals game with $n$ receivers, suppose that sender's strategy is $\Pi=(\pi,Z)$ and other receivers' strategy profile is $\mathbf{a}_{-i}$ in the subgame following $\Pi$. Conditional on having degree $d$ and observing $s \ne \emptyset$, receiver $i$ with type $t$ believes state $\omega=1$ realises with posterior probability
	\[
	\frac{\pi(s|\omega=1)\mu_{t}(\omega=1)}{\pi(s|\omega=1)\mu_{t}(\omega=1) + \pi(s|\omega=0)\mu_{t}(\omega=0)}
	\]
\end{lemma}

The main idea of the proof is that network formation process and seeding strategy do not depend on the state, so the posterior upon observing some nonempty signal is completely determined by the signal structure $\pi$ and prior. 

Some important observations follow from this lemma. If some signal $s$ persuades both types when sent publicly, then all receivers will share it if they receive it. If such a signal is realised, then a receiver will observe it if and only if she is connected to one of the seeds. In other words, all receivers on components\footnote{A component in a network is a set of nodes such that all nodes in this set are connected to each other, and all nodes not in this set are not connected to nodes in the set.} that contain at least one seed will observe the signal. If $s$ persuades only the believers but not the sceptics when sent publicly, then a receiver $i$ observes the signal if and only if there exists a path from $i$ to some seed such that all receivers on that path excluding $i$ are believers. 
In other words, if a component in the subnetwork of believers\footnote{The subnetwork of believers is the network formed by deleting all sceptics and links of sceptics from $g^{n}$} contains a seed, then all believers on this component observe the signal. A sceptic observes it if she has a neighbour who is a believer and observes the signal. 

Receivers' behaviour upon observing non-empty signals is directly implied by this lemma. We still need to understand receivers' posterior upon observing $\emptyset$. To obtain these, we need the probabilities of observing $\emptyset$ in each state. I apply results from random networks to study probabilities of observing non-empty signals. The probability of observing $\emptyset$ can then be obtained as the complementary probability. 

Let $L_{i}(g^{n})$ denote the $i$-th largest component in realised network $g^{n}$, and $|L_{i}(g^{n})|$ denote the number of nodes it contains. Theorem 3.1 and 3.6 in \citet{bollobas2007phase} show that
\begin{thm}[Bollobas et al. (2007)] \label{sadlermain}
    There exists $c \in [0,1]$ such that $|L_{1}(g^{n})|/n$ converges in probability to $c$, and $|L_{i}(g^{n})|/n$ converges in probability to 0 for $i > 1$.
\end{thm}
This theorem says that there is at most one component that is huge, in the sense that its size grows with $n$ linearly. The largest component $L_{1}$ is usually called the \textit{giant component}. 

As argued above, if a signal is only passed on by believers, we need to look at the subnetwork of believers, i.e. the network formed by deleting all sceptics and their links from $g^{n}$. I use $\hat{L}_{i}(g^{n})$ to denote the $i$-th largest component in the subnetwork of believers. The subnetwork of believers depends only on links between believers. Also, conditional on connectedness of a pair of believers, whether they are linked is independent of how sceptics are linked to each other and how sceptics are linked to believers. Therefore, the distribution over links between believers, and thus the subnetwork of believers, will not change if we set $\lambda_{i}$ to be 0 for all sceptics. Applying \Cref{sadlermain} then gives
\begin{cor} \label{cor:giant_in_h_subnetwork}
    There exists $\hat{c} \in [0,1]$ such that $|\hat{L}_{1}(g^{n})|/n$ converges in probability to $\hat{c}$, and $|\hat{L}_{i}(g^{n})|/n$ converges in probability to 0 for $i > 1$.
\end{cor}

Using these results on giant component and \Cref{lemma:non_empty_posterior}, we can specify some conditions under which the receivers observe $\emptyset$ with very high probability in large social networks. The following lemma shows that in large networks, if the seeds are not on the relevant giant component, then the proportion of receivers who observe the signal is very small.
\begin{lemma} \label{lemma:signal_need_to_be_viral}
Fix any $\epsilon>0$.
For $n$ large enough, given any $s \in S$ and any seeding strategy $Z$ such that $Z(g^{n},s) \cap L_{1}(g^{n}) = \emptyset$ for all possible $g^{n}$, the proportion of receivers who observe $s$ is larger than $\epsilon$ with probability at most $\epsilon$. \\ Also, if $s$ is only passed on by believers and $Z(g^{n},s) \cap \hat{L}_{1}(g^{n}) = \emptyset$ for all possible $g^{n}$, the proportion of receivers who observe $s$ is larger than $\epsilon$ with probability at most $\epsilon$. 
\end{lemma}
Intuitively, the lemma says that for a signal to reach many people, it must become viral, in the sense that it spreads on the giant component and thus is shared by a large fraction of the population. The main intuition of the proof is obvious from \Cref{sadlermain}.
Since all components other than the giant component contain a vanishing fraction of nodes, if no seeds hit the giant component, all components hit by the seeds have sizes that are vanishing as a proportion of the total population.

\begin{lemma} \label{lemma:only_observe_viral}
Fix any $\epsilon>0$. 
For $n$ large enough, given any $s \in S$ and any seeding strategy $Z$, the proportion of receivers who are not on $L_{1}(g^{n})$ and observe $s$ is larger than $\epsilon$ with probability at most $\epsilon$. \\Also, if $s$ is only passed on by believers, the proportion of receivers who have no neighbour on $\hat{L}_{1}(g^{n})$ and observe $s$ is larger than $\epsilon$ with probability at most $\epsilon$.
\end{lemma}
Intuitively, this lemma says that receivers can only observe the signal if they are on parts of the network where the signal has gone viral, i.e. the giant component. The main idea of proof is again intuitive. If some receiver is not on the giant component, then its component has vanishing size and is very unlikely to be hit by any seed.

Combining \Cref{lemma:only_observe_viral} and \Cref{lemma:signal_need_to_be_viral}, the probability that a receiver can observe $s$ is approximately the probability that $s$ hits the appropriate giant component, times the probability that the receiver is on (or has a neighbour on) that giant component. The next lemma describes probabilities of being linked to the giant component, conditional on type and degree of a node.

\begin{lemma} \label{lemmaprobongiant}
There exist functions $\zeta(t,d)$ and $\hat{\zeta}(t,d)$ both increasing in $d$ such that, among receivers with degree $d$ and type $t$: \begin{enumerate}
    \item the proportion of those who are on $L_{1}(g^{n})$ converges in probability to $\zeta(t,d)$.
    \item 
 the proportion of those who have a neighbour on $\hat{L}_{1}(g^{n})$ converges in probability to $\hat{\zeta}(t,d)$.
\end{enumerate}  
\end{lemma}
The functions $\zeta$ and $\hat{\zeta}$ characterise limiting conditional probabilities of being on or having a neighbour on the giant components $L_{1}(g^{n})$ and $\hat{L}_{1}(g^{n})$. The equations which characterise functions $\zeta$ and $\hat{\zeta}$ are not key to the results below, so are only presented in the appendix.

\paragraph{Sender's optimal strategy:} Characterisation of optimal strategies is not required for the main results below, and also does not provide much insight and intuition for them. Therefore, a brief discussion of them is relegated to \Cref{sec:extension} and the formal details to the online appendix.

\section{Results}\label{sec:results}

 \subsection{Baseline network}
 In this subsection, I study the simplest case where the network formation process ignores the types of receivers. Formally, I study the case where distribution of connectedness $f_{t}$ is independent of $t$ and $q=1$.  
 The next lemma provides part of the main intuition underlying this subsection. 
 \begin{lemma} \label{lemma:probability_on_giant_baseline}
 	If $f_{l}(\lambda)=f_{h}(\lambda)$ for all $\lambda$ and $q=1$, then we have $\zeta(h,d)=\zeta(l,d)$ and $\hat{\zeta}(h,d)=\hat{\zeta}(l,d)$ for all $d$.
 \end{lemma}
 This follows from the fact that, in this baseline network, the network formation process essentially ignores types of receivers. Intuitively, this says that if a signal gets viral, then believers and sceptics have roughly the same probability of receiving it when $n$ is large. 
\begin{proposition} \label{baselineprop}
In networks without homophily or correlation between degree and type, a risk-neutral or risk-loving sender prefers public signals. \\
Formally, in networks where $f_{l}(\lambda)=f_{h}(\lambda)$ for all $\lambda$ and $q=1$, if the sender's payoff function $v$ is linear or convex, then $\limsup_{n \rightarrow \infty} V^{*}_{net}(n) \le V^{*}_{pub}$.
\end{proposition}

As discussed in introduction, with heterogeneity of priors, the sender wants to provide information only to the sceptics but not to the believers. For instance, if a signal persuades both types of receivers, ideally the sender wants believers not to observe it, and sceptics to observe it very often. Is there a way of choosing seeds to achieve this using network signals?

By \Cref{lemma:signal_need_to_be_viral}, for a signal to reach a large fraction of the population, we need the signals to be viral, i.e. the seeds of the signal needs to be on the relevant giant component. However, once the signals become viral, it is difficult to control who the signal reaches. In particular, in the baseline network, \Cref{lemma:probability_on_giant_baseline} states that once a signal becomes viral, it must reach believers and sceptics with the same probability. Therefore, network signals cannot make sceptics observe informative signals more often than believers. 

\citet{arieli2019private} finds that when receivers are homogeneous, public signals are optimal for the sender. However, in sharp contrast with \Cref{baselineprop} above, they also find that when the sender's payoff function is strictly increasing\footnote{Their condition is actually weaker. They only need the sender to prefer 1 receiver choosing action 1 to no receiver choosing action 1}, public signals cannot be optimal if receivers are heterogeneous, while the sender prefers public signals even with heterogeneous priors in \Cref{baselineprop}. The difference is driven by \citet{arieli2019private} having no communication between receivers. Communication between receivers impose restrictions on correlations of observations between different receivers. \Cref{lemma:probability_on_giant_baseline} sheds light on these restrictions, and intuition above explains why these restrictions imply that network signals cannot outperform public signals. 

I provide a rough sketch of the proof for \Cref{baselineprop}. Suppose for simplicity that the sender's network signal structure $\pi$ uses only 2 non-empty signals $s$ and $s'$, where $s$ persuades both types of receivers to take action 1 and $s'$ persuades only the believers, i.e. $\frac{\pi(s|\omega=1) }{\pi(s|\omega=0) } > \frac{\mu_{l}(\omega=0) }{\mu_{l}(\omega=1) }  $, and $\frac{\mu_{h}(\omega=0) }{\mu_{h}(\omega=1) }< \frac{\pi(s'|\omega=1) }{\pi(s'|\omega=0) } < \frac{\mu_{l}(\omega=0) }{\mu_{l}(\omega=1) }  $.

By \Cref{lemma:signal_need_to_be_viral}, if no seeds are on the relevant giant component, then very few people will observe it. Therefore, for the sketch here, suppose that the seeds for $s$ are always on $L_{1}(g^{n})$ and the seeds for $s'$ are always on $\hat{L}_{1}(g^{n})$. For receivers with degree $d$ and type $t$, let $a^{*}(\emptyset;t,d,\Pi)$ denote their action upon observing $\emptyset$ in equilibrium of the subgame following sender's strategy $\Pi$.

When $n$ gets large, roughly a proportion $\gamma_{t}p_{d}^{t}$ of the population has type $t$ and degree $d$. Among them, roughly a proportion $\zeta(t,d)$ is in the giant component $L_{1}$. Therefore, when $s$ is realised, roughly $\zeta(t,d)$ of those with type $t$ and degree $d$ observe $s$, and others choose $a^{*}(\emptyset;t,d,\Pi)$. Also, roughly a proportion $\hat{\zeta}(t,d) $ of them have a neighbour on the giant component in the subnetwork of type $h$ nodes. Therefore, when $s'$ is realised, roughly $\hat{\zeta}(t,d)$ of those with type $t$ and degree $d$ observe $s'$, and others choose $a^{*}(\emptyset;t,d,\Pi)$. By \Cref{lemma:probability_on_giant_baseline}, $\zeta(t,d)$ and $\hat{\zeta}(t,d)$ do not depend on $t$, so I omit the argument $t$ in this sketch. Also, $p_{d}^{h}$ and $p_{d}^{l}$ are the same in the baseline network, so I omit the superscript. 
Using these, sender's payoff from $\Pi$ is close to
\begin{multline}
  \sum_{\omega}\mu_{s}(\omega) \Bigg[ \pi(s|\omega)v(\sum_{d}p_{d}\sum_{t}\gamma_{t}\big[ \zeta(d) +  (1-\zeta(d)) a^{*}(\emptyset;t,d,\Pi)   \big]    )  \\    +  \pi(s'|\omega)v(\sum_{d}p_{d}\big[ \gamma_{h}\hat{\zeta}(d) + (1-\hat{\zeta}(d))\sum_{t}\gamma_{t} a^{*}(\emptyset;t,d,\Pi)  \big]    )   + \pi(\emptyset|\omega)v( \sum_{d}p_{d} \sum_{t}\gamma_{t} a^{*}(\emptyset;t,d,\Pi) ) \Bigg]  
\end{multline}
If $v$ is convex, then this is smaller than $\sum_{d}p_{d} V(d,\Pi)$ defined as
\begin{multline}
   \sum_{d}p_{d} \sum_{\omega}\mu_{s}(\omega)\Bigg[ 
  \pi(s|\omega)\zeta(d)v(1) + \pi(s'|\omega)\hat{\zeta}(d)v(\gamma_{h}) \\ + (1- \pi(s|\omega)\zeta(d) - \pi(s'|\omega)\hat{\zeta}(d))v(\sum_{t}\gamma_{t} a^{*}(\emptyset;t,d,\Pi) ) 
  \Bigg]  \equiv \sum_{d}p_{d} V(d,\Pi)
\end{multline}
Therefore, to show that public signals are better, it's sufficient to show that we can imitate the payoffs $V(d,\Pi)$ for each $d$. Consider a public signal strategy that sends the signal $s$ with probability $\pi(s|\omega)\zeta(d)$ and $s'$ with probability $\pi(s'|\omega)\hat{\zeta}(d)$ in each state $\omega$, and sends $\emptyset$ with the remaining probability. Note that these are roughly the probabilities that agents with degree $d$ observe $s$, $s'$ and $\emptyset$ under the network signal strategy $\Pi$, therefore posteriors for each observation are also close to those under $\Pi$. A conjecture is that this implies that receiver's action given each observation will be the same as those with degree $d$ under the network signal strategy $\Pi$. A technical difficulty is that receiver's action changes discontinuously at posterior 0.5, so posterior being close does not directly imply that they choose the same action. 
This is dealt with in the proof, and here for the sake of illustration, let us take this conjecture to be true. Given this conjecture, one can check that the sender's payoff from this public signal structure is $V(d,\Pi)$.

Intuitively, given any $d$, the sender can use public signals to imitate the probabilities that receivers observe each signal under some network signal structure. This is possible because, conditional on degree, probabilities of observing signals are (roughly) the same for both types of receivers.

\subsection{Homophily and type-dependent degree distribution}
In this subsection, I study networks that could have homophily and correlation between degrees and types. The conditions for the results in this subsection will be stated in terms of $\{p_{d}^{t}\}_{d}$ and $q_{t}$, i.e. the degree distribution and expected proportion of neighbours of the same type, instead of the distribution over $\lambda$ and $q$. There are two reasons for this. Firstly, from data, it is relatively easy to calculate the degree distribution and $q_{t}$, but it is difficult if not impossible to estimate $q$ and the distribution of $\lambda$. Therefore, stating the results in terms of the degree distributions and $q_{t}$ makes it easier to check whether the conditions of the propositions are satisfied. Secondly, all the results would go through when using the more general random networks model from \citet{sadler2020}, which directly imposes arbitrary degree distributions and values of $q_{t}$, so the results are still relevant even if we face some degree distribution that could not be generated by the more restrictive model used here. I also discuss relations between degree distribution, $q_{t}$ and the distribution over $\lambda$ and $q$ later in this subsection. 

I need to introduce some definition to state the results. First, I define the forward (degree) distribution. The forward distribution captures the degree distribution of a randomly selected neighbour. 
\begin{definition}
The forward distribution $\{\tilde{p}_{d}^{t}\}_{d}$ for type $t$ is defined as $\tilde{p}_{d}^{t} \equiv  \frac{p_{d+1}^{t}(d+1) }{\sum_{d}d p_{d}^{t}  }$ for all $d$.  
\end{definition}
Next I define what is meant by one type being more connected than the other.
\begin{definition}
    Type $t$ receivers are weakly more connected than type $t'$ if one of the two following conditions holds \begin{enumerate}
        \item The degree distributions $\{p_{d}^{t}\}_{d}$ and $\{p_{d}^{t'}\}_{d}$ are the same.
        \item $\{p_{d}^{t}\}_{d}$ FOSD (first order stochastically dominates) $\{p_{d}^{t'}\}_{d}$. If $q_{t} \ne 1 - q_{t'}$, then we also require forward distribution of type $t$ to FOSD forward distribution of type $t'$. 
    \end{enumerate}
\end{definition}
In words, type $t$ is weakly more connected to type $t'$ if they have the same degree distribution, or the degree distribution and forward distribution of type $t$ FOSD those of type $t'$ respectively. 
Like in the last subsection, we first study the probability that receivers observe viral signals. 
\begin{lemma} \label{lemma:probability_on_giant_homophily}
	If $q_{h}$ is (strictly) larger than $1-q_{l}$, then $\hat{\zeta}(h,d)$ is (strictly) larger than $\hat{\zeta}(l,d) $ for all $d$. 
    If type $h$ receivers are weakly more connected than type $l$ receivers, then $\zeta(h,d) \ge \zeta (l,d)$ for all $d$.
\end{lemma}

Before stating the result, I briefly comment on how $q_{t}$ is related to homophily. \citet{coleman1958relational}'s measure of homophily in this model is $\frac{q_{t}-\gamma_{t} }{1- \gamma_{t}}$, and there is homophily of type $t$ if this measure is positive. Therefore, the network exhibits homophily if $q_{t} > \gamma_{t}$ for both types. Note that if $q_{h} > \gamma_{h}$ and $q_{l} > \gamma_{l}$, then we must have $q_{h} > 1-q_{l}$ since $\gamma_{h}$ and $\gamma_{l}$ sum up to 1. Therefore, the condition in lemma \ref{lemma:probability_on_giant_homophily} and the proposition below are weaker than requiring the network to exhibit homophily. 

Using lemma \ref{lemma:probability_on_giant_homophily}, we can show that conclusions from \Cref{baselineprop} extend even when the network exhibits homophily, if believers are weakly more connected than sceptics.  
\begin{proposition} \label{prop:main_prop}
In networks where type $h$ receivers are weakly more connected than type $l$ receivers and $q_{h} \ge 1-q_{l}>0$, a risk-loving or risk neutral sender prefers public signals. \\
Formally, in such networks, if $v$ is linear or convex, then $\limsup_{n \rightarrow \infty} V^{*}_{net}(n) \le V^{*}_{pub}$.
\end{proposition}

It should be emphasized that when setting up the random network model, complete segregation is ruled out by setting $q>0$. This implies that the requirement $1-q_{l}>0$ in the proposition always holds. Note that there can still be very high levels of homophily, i.e. $q_{h}$ and $q_{l}$ can both be arbitrarily close to 1. When there is complete segregation, the sender is obviously in the ideal situation where she could send information only to the sceptics and not to the believers. 

Homophily is often cited as one of the reasons why advertising and political campaigns exploiting social networks are more effective than traditional forms of advertisement, because it influences what information people could observe. 
\Cref{lemma:probability_on_giant_homophily} indeed shows that homophily and type-degree correlation introduce asymmetries between the two groups. Roughly speaking, the sender would like to send signals more often to sceptics and keep information away from believers. From \Cref{lemma:probability_on_giant_homophily}, when there is homophily and believers are weakly more connected in the network, the asymmetry introduced is the opposite of what the sender desires.

To discuss the mechanism in more detail, I again use the simplified setup for the sketch proof of \Cref{baselineprop}, where the sender uses only two signals $s$ and $s'$, with $s$ persuading both types and $s'$ persuading believers only. 
 From \Cref{lemma:probability_on_giant_homophily} and focussing on the case where the two types are equally connected, homophily makes believers observe signal $s'$ more often than sceptics. This asymmetry means that the imitation argument from last subsection does not directly work. Nonetheless, we could still imitate the probabilities of observing each signal for receivers of a given type and degree $d$. For any degree $d$, there are two cases. If sceptics with degree $d$ choose action 0 upon observing $\emptyset$, then they only choose action 1 when they see $s$. Therefore, if we let them observe signals with the same probabilities as believers with degree $d$, then their probabilities of choosing action 1 will not change. In the second case, suppose that sceptics choose action 1 upon observing $\emptyset$, then believers must always choose action 1 when they observe signals with the same probabilities as sceptics since they have higher priors. Therefore, when the sender has weakly convex payoff functions, imitating probabilities for one of the types gives the sender higher payoff than her 'payoff' from type $d$ receivers (i.e. counterpart here for $V(d,\Pi)$ in last subsection). 

Now, if believers are in addition more connected, then they will also observe $s$ with higher probabilities. This is exactly the opposite of what the sender desires. The sender would like sceptics to observe $s$ since they need convincing, and keep $s$ away from believers since observing $s$ too often can make $\emptyset$ too informative of the bad state.

This result does not say that homophily is not useful for the sender when she has to use network signals. 
Homophily makes believers more clustered together, and signals that only persuade believers can spread more widely. Compared to a network without homophily, sender could send these signals less often while still maintaining the same probability that a believer with degree $d$ observes them. These signals are realised less often now, but are observed by more receivers when realised, which increases risk and potentially makes the sender better off.

The next remark relates the degree distributions and $q_{t}$ to the distributions of $\lambda$ and $q$.
\begin{rem} \label{rem:conditions}
\begin{enumerate}
    \item $q_{h} > 1- q_{l}$ if and only if $q<1$, and $q_{h} = 1-q_{l}$ if and only if $q=1$.
    \item If $q=1$, type $t$ is weakly more connected than type $t'$ if distribution of $\lambda$ does not depend on type or if $f_{t}$ FOSD $f_{t'}$ 
    \item Given any $f_{t'}$, we can find some $r$ large enough such that type $t$ is more connected than type $t'$ if $f_{t}(r \lambda) = f_{t'}(\lambda)$ for all $\lambda$.\\ A special case is the island model often used in literature where all receivers have the same $\lambda$. In these models, the larger group is more connected. 
\end{enumerate}
\end{rem}

Next, I present conditions under which information can be sent to sceptics more often than believers, and therefore network signals are better than public signals for the sender. 
\begin{proposition} \label{prop:sceptics_more_connected}
If expected degree of sceptics is sufficiently large relative to expected degree of believers, the sender prefers network signals.\\
   Formally, suppose that there's no type $l$ receiver with $\lambda$ smaller than some arbitrary $\lambda' >0$. We can find some $\overline{x}$ large enough such that if expected degree of type $l$ receivers is larger than $\overline{x}$ and expected degree of type $h$ receivers is smaller than $1 / \overline{x}$, then $\liminf_{n \rightarrow \infty} V_{net}^{*}(n) > V_{pub}^{*}$.
\end{proposition}
Intuitively, believers are less connected than sceptics, and therefore they receive information less often than sceptics, which is desirable for the sender. As an extreme example, consider a network where all believers are isolated and all sceptics are connected together. This allows the sender to achieve the ideal outcome of sending information only to sceptics and keeping believers completely uninformed. 
\paragraph{Discussion}
One high level intuition underlying this and the last subsection is that, it is difficult to control who sees viral signals using seeding. It is the network structures and how people share signals that determine who observe viral signals. Therefore, whether a risk-loving or risk-neutral sender prefers network signals depends, roughly speaking, on whether the network structure causes believers to receive viral signals less often than sceptics.

As discussed in the literature review, most existing papers on persuasion do not consider endogenous communication between receivers on a network. 
Communication and the networks impose endogenous restrictions on correlations between what the receivers observe. Lemmas \ref{lemma:probability_on_giant_baseline} and \ref{lemma:probability_on_giant_homophily} shed light on these restrictions, and propositions \ref{baselineprop} to \ref{prop:sceptics_more_connected} characterise networks in which public signals can and cannot be outperformed for a risk-loving or risk-neutral sender.

In particular, propositions \ref{baselineprop} and \ref{prop:main_prop} show that these restrictions can severely constrain the sender's payoff using network signals compared to public signals. If the sender can use public signals and knows that network satisfy conditions in these propositions. then more detailed information about the network is not very valuable to the sender. In practice, sending public signals may be more costly than using network signals. However, it is also costly to collect and analyse detailed information about network structure, especially in large and complex networks. Therefore, these results imply that it may be more efficient for the sender to spend resources on sending public signals rather than collecting information about the network. 

When sceptics are more connected than believers, \Cref{prop:sceptics_more_connected} implies that there is a stronger case for understanding and exploiting the network. In this case, the sender has incentive to know how to seed signals in order to make it go viral (i.e. to identify the giant component). Such information allows the sender to get higher payoff than using public signals. 

\subsection{Risk aversion} \label{subsec:risk_aver}
Results in the last two subsections are mainly about risk-loving and risk-neutral senders, except for \Cref{prop:sceptics_more_connected} which puts no condition on the sender's preference in addition to those in model setup. If the sender is not very risk-averse (i.e. payoff function $v$ is concave but very close to linear), results above still apply. This subsection considers a sufficiently risk-averse sender. 
\begin{proposition} \label{prop:risk_aversion}
    If $\hat{c}>0$ (i.e. the giant component in the subnetwork of believers is non-vanishing), there exists a concave payoff function $v$ such that $\liminf_{n \rightarrow \infty} V^{*}_{net}(n) > V^{*}_{pub}$.\\ In particular\footnote{I thank an anonymous referee who suggested this example as a conjecture}, suppose that $v(x) = (x^{b}-1)/b$, which has constant relative risk aversion. If $\hat{c}>0$, the sender prefers network signals for sufficiently small $b$.
\end{proposition}
A necessary and sufficient condition for $\hat{c}>0$ can be found in \citet{bollobas2007phase}. A sufficient condition is that $\lambda_{min}^{2}\gamma_{h}\sum_{\lambda>0}f_{h}(\lambda) >1 $, where $\lambda_{min}$ is the minimum non-zero value $\lambda$ could take. Intuitively, this condition says that believers are sufficiently connected. 

With network signals, receivers of the same type can have different observations. When $\hat{c}>0$, by sending signals that persuade only believers to the giant component in the subnetwork of believers, the sender can persuade all believers and at least some sceptics all the time. However, it is in general impossible to persuade all sceptics at the same time using network signals. With public signals, all receivers of the same type must choose the same action, so it is impossible to persuade some sceptics all the time. A sender who is very risk-averse prefers network signals, since it involves less risk.

\subsection{Voting game} \label{subsec:vote}
An often studied persuasion game is the voting game. In voting games, the sender only cares about whether the proportion of receivers taking sender's preferred action reaches a certain threshold. 
Let the sender's payoff function be defined in the following way: There exists some $\overline{x} < 1$ such that $v(x)=1 $ if $x \ge \overline{x}$ and $v(x)=0$ otherwise. To make the problem non-trivial, I assume that $\overline{x} > \gamma_{h}$. 

To state the result, I introduce the following definition. Define $d_{vote}$ to be 
	\[
	d_{vote} \equiv \min\{ d \in \mathbb{N}_{0}: \frac{\mu_{l}(\omega=1)}{\mu_{l}(\omega=0)} \frac{ 1 - \frac{\mu_{h}(\omega=0)}{\mu_{h}(\omega=1) }\hat{\zeta}(l,d) }{1 - \hat{\zeta}(l,d) } \ge 1\}
	\]
	Let $d_{vote}$ be $\infty$ if there exists no such $d$. As shown in proof of the next proposition, if a type $l$ receiver has degree lower than $d_{vote}$, then she cannot choose action 1 upon observing $\emptyset$. Using this definition, we can provide a necessary condition and a sufficient condition for network signals to outperform public signals. 
	\begin{proposition} \label{propvotinggame}
		A necessary condition for $\limsup_{n \rightarrow \infty} V_{net}^{*}(n) > V_{pub}^{*}$ is
		\[
		\sum_{d \ge d_{vote}} \gamma_{l} p_{d}^{l} \ge \overline{x} - \gamma_{h}
		\]
		A sufficient condition for $\liminf_{n \rightarrow \infty} V_{net}^{*}(n) > V_{pub}^{*}$ is
		\[
		\sum_{d \ge d_{vote}+1} \gamma_{l} p_{d}^{l}(1-\hat{\zeta}(l,d)) > \overline{x} - \gamma_{h}
		\]
	\end{proposition}
    The following corollary provides conditions on the network that are more intuitive than those in \Cref{propvotinggame}. It highlights the importance of the giant component in the subnetwork of believers. In particular, it provides a condition on the network such that the sender prefers public signals regardless of the threshold $\overline{x}$
    \begin{cor} \label{cor:votinggame}
        If $\hat{c}=0$ (i.e. the giant component in the subnetwork of believers is vanishing), then the sender prefers public signals. If $\hat{c}>0$, then the sender prefers network signals if $\overline{x}-\gamma_{h}$ is sufficiently small.
    \end{cor}
As explained in the last section, whether $\hat{c}$ is positive or not intuitively depends on whether the believers are sufficiently connected in the network. 
    
The intuition and driving forces underlying these results are actually similar to those for \Cref{prop:risk_aversion} on a risk-averse sender. Both conditions in \Cref{propvotinggame} require $\overline{x} - \gamma_{h}$ to be small. When this difference is small, the sender is quite risk averse in the sense that her payoff increases sharply by getting a bit more support in addition to that from believers, but does not increase beyond that. The left hand sides for both conditions in \Cref{propvotinggame} capture the proportion of sceptics that can be made to always choose action 1. Therefore, both conditions in \Cref{propvotinggame} say that when the sender is very 'risk-averse', the sender benefits from the ability to always persuade some sceptics using network signals. This is also what drives \Cref{prop:risk_aversion}. Finally, \Cref{prop:risk_aversion} requires $\hat{c}>0$ for a risk-averse sender to prefer network signals, and the corollary above shows that $\hat{c}=0$ is sufficient for network signals to perform worse in voting games. 

When the sender is very risk-averse or when $\overline{x} - \gamma_{h}$ is small, a main factor influencing whether the sender prefers network signals is whether the sender can always persuade all believers and some sceptics. To achieve this, the sender needs to let sceptics observe signals that persuade only the believers. How well the sender can do this in turn depends on whether such signals can get viral among believers (i.e. whether $\hat{c}>0$). This is why $\hat{c}$ is a key measure of the network in this and the last subsection. 

Recall that when the sender is risk-loving or risk neutral, what matters is whether sceptics are more connected than believers in the network. We can see that the sender's preference can have important influence on the aspect of the network that matters for comparison between public signals and network signals.

\section{Extensions}\label{sec:extension}
\paragraph{Alternative sharing rules:}
As discussed above in detail, the intuition behind propositions \ref{baselineprop} and \ref{prop:main_prop} is that believers observe signals more often than sceptics. These results would still go through with any other rule where (conditional on degree) believers share each signal $s$ more often than sceptics. Such sharing rules ensure that, under the conditions stated in the propositions \ref{baselineprop} and \ref{prop:main_prop}, believers observe each signal more often than sceptics. 

I now discuss one sharing rule that does not fall into this class, but the results still go through. Suppose that now receivers share a signal $s$ if their action agrees with the content of the signal. Formally, they share the signal if their action is 1 and the signal is more likely in state 1, or if their action is 0 and the signal is more likely in state 0. All receivers pass on signals that persuade both types or neither type. Only sceptics pass on signals that persuade only the believers, so sceptics can observe these signals more often than believers even if believers are weakly more connected.

However, this cannot make network signals better than public signals if the believers are more connected. To see this, one can go back to the discussion after proposition \ref{prop:main_prop}, explaining why public signals are still preferable when network exhibits homophily. The argument there does not use the fact that signals persuading only believers are observed more often by believers. More formally, \Cref{lemmaboundvhat} in the appendix, which presents the formal version of imitation arguments in the main text, does not require signals persuading only believers to be observed more often by believers. Therefore, propositions \ref{baselineprop} and \ref{prop:main_prop} go through in essentially the same way with this alternative sharing rule. 

The driving force for proposition \ref{prop:sceptics_more_connected} is unchanged by these alternative sharing rules. 
For the case of a risk-averse sender and voting, the exact network measures that matter can be different from the ones in subsections \ref{subsec:risk_aver} and \ref{subsec:vote}. Nonetheless, the high level intuition that the sender cares about always persuading all believers and some sceptics still holds. Therefore, results analogous to those in subsections \ref{subsec:risk_aver} and \ref{subsec:vote} can be obtained. 

\paragraph{Public signals do not reach all receivers:}
In the analysis above, it is assumed that public signals will reach everyone. In practice, public broadcasts may not reach all the audience.  

Suppose that now the realised public signal will be observed by a proportion $y$ of uniformly randomly selected receivers, and the other receivers observe $\emptyset$. Propositions \ref{baselineprop} through \ref{prop:sceptics_more_connected} still hold if $y$ is large relative to size of the giant component.
Detailed proofs follow very similar steps, so will be omitted for the sake of space. I illustrate the main idea by discussing when and why the sketch proof of proposition \ref{baselineprop} still goes through. Recall that the sketch proof works by using the public signals to imitate probabilities that receivers with degree $d$ observe $s$ and $s'$. The public signal should also send $s$ and $s'$ with probabilities smaller than those of sending $s$ and $s'$ in the network signal structure\footnote{This also implies that each of these signals will reach more than $\zeta(d)$ and $\hat{\zeta}(d)$ proportions of receivers with degree $d$}. 
Using network signals, these receivers observe $s$ with probability $\pi(s|\omega)\zeta(d)$ and $s'$ with probability $\pi(s'|\omega)\hat{\zeta}(d)$. In the current setup, we need the public signal structure to send $s$ and $s'$ with probabilities $\pi(s|\omega)\zeta(d)/y$ and $\pi(s'|\omega)\hat{\zeta}(d)/y$ respectively. For the imitation argument to work, we need $\zeta(d)/y$ and $\hat{\zeta}(d)/y$ to be less than 1. 

Note that $\zeta(d)$ and $\hat{\zeta}(d)$ go to 1 as $d$ goes to infinity, so the imitation argument does not work for high degrees. However, the proportion of nodes with high degrees is small, so their actions do not matter much for the sender's payoff. 

The argument above shows why propositions \ref{baselineprop} and \ref{prop:main_prop} still hold for a risk-loving or risk-neutral sender. The intuition for why the sender can prefer network signals when she is risk averse (and $y$ is large enough) or when sceptics are more connected are essentially the same as in the main text. 


\paragraph{Non-vanishing fraction of seeds:}
The assumption of seeds being a vanishing fraction of population can be relaxed as long as the seeds only take a fraction of population that is small enough. An upper bound on sender's payoff can be obtained by assuming that all the seeds and others on the same component with them, excluding those on the giant component, always choose action 1 and remove them from the network. We can then analyse the remaining network and the results above will go through. This implies that results would go through as long as the fraction of seeds is small enough even though it is non-vanishing. 

On the other hand, if the fraction of seeds is larger than $\gamma_{l}$, i.e. the fraction of sceptics. Then the sender could always just directly send signals that only realise in state 0 to all the sceptics, to persuade them to take action 1 upon observing $\emptyset$. Believers will receive no information. Sender achieves her highest possible payoff (when $v$ is linear or convex) by doing this. 

For intermediate sizes of seed set, we need to consider components other than the giant component. We need to know the probabilities that a node with type $t$ and degree $d$ is on components of all possible sizes and compositions. As $n$ grows large, the number of probabilities that we need to keep track of for each $t$ and $d$ also goes to infinity. This makes analysis a lot more complicated. It can be shown that, if the sender chooses seeds uniformly randomly, then proposition \ref{baselineprop} still goes through. In general, it is difficult to get tractable results.

\paragraph{Sender's optimal strategy:}
\citet{innocenti2022can} characterises optimal public signals with heterogeneous receivers, where the sender's payoff function is linear. It can be shown that the characterisation still holds for payoff functions used in this paper. 

Optimal network signal strategies are very complex mainly for two reasons. First, even though the probabilities that receivers are on giant component (i.e. functions $\zeta$ and $\hat{\zeta}$) can be characterised, there exist no closed-form solutions for them. Secondly, receivers' optimal actions and thus sender's payoff are discontinuous at the point where receivers have a posterior of 0.5. Therefore, when receivers' posterior converge to 0.5 as $n$ gets large, their actions and the sender's payoff may not converge.

In the online appendix, I show that under certain parameter conditions, the problem of finding the sender's limiting optimal payoff can be reduced to a problem where the sender is choosing probabilities of two non-empty signals. However, outside these parameters, the sender's problem can be much more complex. Lemma \ref{lemma:signal_need_to_be_viral} says that seeds not on the relevant giant component will have little influence. It is also confirmed by the results in section \ref{sec:results} that they do not matter for comparison with public signals. Nevertheless, the online appendix illustrates that these seeds with little influence could be very important for the sender's optimal payoff using network signals. This is due to the second point in the last paragraph.

\section{Conclusion}\label{sec:conclusion}
There has been a lot of discussion on exploiting social networks to spread information and persuade people in political and advertising campaigns. This paper studies whether and when network signals outperform public signals. This paper finds that a risk neutral or risk-loving sender prefers public signals if believers are more connected than sceptics. This holds even if the network exhibits homophily. The sender prefers network signals if sceptics are sufficiently more connected than believers. When believers are sufficiently connected, a very risk averse sender prefers network signals. If the sender's payoff function is a threshold function, connectedness of believers is also important for the comparison. 

The main intuition underlying results in this paper is that, the sender needs the signal to get viral in order to have an impact. However, it is difficult to control who can observe viral signals using seeding. The results specify how different networks spread viral signals, and whether the sender can exploit it to outperform public signals. 

In many situations, the campaigner has access to tools that can send public signals. The results in this paper imply that, in many cases, resources may be more efficiently spent on sending public signals, rather than on trying to understand network structures and picking the optimal seeds. Conditions under which network signals perform better are also specified, and there is a stronger case to collect more information about the network under these conditions. 
In addition, persuasion on large networks is complicated, while public signals have been well studied. From these comparisons, we can use sender's value from public signals to obtain what sender could get at most or at least using network signals. Even if the sender has no access to public signals, the results could be useful to judge whether a campaign using network signals will be profitable for the sender.

\appendix
\numberwithin{equation}{section}

\section{Appendix: Proofs}\label{sec:proofs}
I use $PO(d,x)$ to denote the probability that a Poisson random variable with intensity $x$ is equal to $d$. I define $D(t,\lambda) \equiv \lambda (\gamma_{t}E_{t}(\lambda) +q \gamma_{t'} E_{t'}(\lambda)) $ where $t' \ne t$. This captures the expected degree of node with type $t$ and connectedness $\lambda$. The next subsection proves the main results, and the subsection after that proves the technical lemmas in \Cref{sec:analysis}.
    \subsection{Proof of main results}
	Using the lemmas \ref{lemma:signal_need_to_be_viral}, \ref{lemma:only_observe_viral}  and \ref{lemmaprobongiant}, I can now approximate probabilities that receivers observe an non-empty signal $s \ne \emptyset$. 
	Let $C(.)$ and $\hat{C}(.)$ be functions over the set of seeds and networks defined in the following way
	\[
	C(z,g^{n}) = \mathbf{1}(z \cup L_{1}(g^{n}) \ne \emptyset  ) 
	 \qquad
	\hat{C}(z,g^{n}) = \mathbf{1}(z \cup \hat{L}_{1}(g^{n}) \ne \emptyset) 
	\]
    In words, $C(z,g^{n})$ is 1 if seeds are on giant component, and $\hat{C}(z,g^{n})$ is 1 if seeds are on giant component after deleting type $l$ nodes. 
	Define $\mathbf{A}^{n}(\epsilon, \overline{d})$ to be the set of networks with $n$ receivers that satisfy the following properties: Given any sender's strategy $\Pi^{n}$ in the network signal game with $n$ receivers, we have
	\begin{enumerate}
		\item Let $s \ne \emptyset$ be a signal that persuades both types of receivers according to $\Pi^{n}$. Conditional on $g^{n} \in \mathbf{A}^{n}(\epsilon,\overline{d})$, signal $s$ and seeds $z$ being realised, the proportion of receivers with type $t$ and degree $d \le \overline{d}$ who observe $s$ is at most $\epsilon$ from $C(z,g^{n})\zeta(t,d)$.
		\item Let $s'$ be a signal that persuades only type $h$ receivers according to $\Pi^{n}$. Conditional on $g^{n} \in \mathbf{A}^{n}(\epsilon,\overline{d})$, signal $s'$ and seeds $z$ being realised, the proportion of receivers with type $t$ and degree $d \le \overline{d}$ who observe $s'$ is at most $\epsilon$ from $\hat{C}(z,g^{n})\hat{\zeta}(t,d)$.
		\item If $s''$ is a signal that persuades no receiver according to $\Pi^{n}$, the proportion of receivers who observe $s''$ is at most $\epsilon$.
	\end{enumerate}
	\begin{lemma} \label{lemma:observesignal}
		Given any $\overline{d} \in \mathbb{N}$ and $\epsilon > 0$, there exists an $n_{1}$ large enough such that for any $n \ge n_{1}$, $Pr(g^{n} \in \mathbf{A}_{n}(\epsilon,\overline{d} )   ) > 1-\epsilon$
	\end{lemma}
	Intuitively, this lemma states that as $n$ gets large, the probability that a receiver observes a certain signal is given by the probability that it hits and the receiver is on the relevant giant component. 
	\begin{proof}
		By lemmas \ref{lemma:only_observe_viral} and \ref{lemmaprobongiant}, for any $\epsilon>0$, we can find $n$ large enough such that with arbitrarily high probability, the proportion of receivers with type $t$ and degree $d$ who observe $s$ conditional on seeds being on relevant giant component is within
		$
		[ \zeta(t,d)+2\epsilon, \zeta(t,d)- 2\epsilon  ]
		$
		if $s$ persuades both types of receivers, and is within
		$
		[ \hat{\zeta}(t,d)+2\epsilon, \hat{\zeta}(t,d)- 2\epsilon   ]
		$
		if $s$ persuades only type $h$ receivers. 
		Result then follows from definition of functions $C(.)$ and $\hat{C}(.)$, and lemma \ref{lemma:signal_need_to_be_viral}.
	\end{proof}
    Lemma \ref{lemma:probability_on_giant_baseline} and proposition \ref{baselineprop} are special cases of lemma \ref{lemma:probability_on_giant_homophily} and proposition \ref{prop:main_prop}, so I only prove the latter two. 

    \textit{Proof of lemma} \ref{lemma:probability_on_giant_homophily}: \\
The first part of the result directly follows from the proof of lemma \ref{lemma:probability_on_hatL}. 

For the second part, by definition, we have that $\rho(\infty|t,\lambda) = \sum_{d} Poisson(d;D(t,\lambda)) \rho(\infty|t,d) $
From remark \ref{rem:what_is_rho}, we have
        \begin{equation} \label{eq1_lemma_prob_on_giant_homo}
            1- \rho(\infty|t,d) = (q_{t} \frac{\sum_{\lambda'}\lambda' f_{t}(\lambda')(1-\rho(\infty|t,\lambda')) }{\sum_{\lambda'}\lambda' f_{t}(\lambda') }  + (1-q_{t})\frac{\sum_{\lambda'}\lambda' f_{t'}(\lambda')(1-\rho(\infty|t',\lambda')) }{\sum_{\lambda'}\lambda' f_{t'}(\lambda') }  )^{d}
        \end{equation}
         
        This directly implies the result when $q_{h} = 1-q_{l}$. I now discuss the case with $q_{h} \ne 1-q_{l}$. The forward distribution for type $t$ is given by
        \[
        \frac{  (d+1)\sum_{\lambda}f_{t}(\lambda)Poisson(d+1; \lambda (\gamma_{t}E_{t}(\lambda) + q \gamma_{t'}E_{t'}(\lambda)   ) )   }{ \sum_{\lambda}  \lambda f_{t}(\lambda)(\gamma_{t}E_{t}(\lambda) + q \gamma_{t'}E_{t'}(\lambda)   )  }
        \]
        which by the probability mass function of Poisson distribution is equal to
\[
        \frac{  \sum_{\lambda}\lambda f_{t}(\lambda)Poisson(d; \lambda (\gamma_{t}E_{t}(\lambda) + q \gamma_{t'}E_{t'}(\lambda)   ) )   }{ \sum_{\lambda}  \lambda f_{t}(\lambda)  } = \frac{  \sum_{\lambda}\lambda f_{t}(\lambda)Poisson(d; D(t,\lambda) )   }{ \sum_{\lambda}  \lambda f_{t}(\lambda)  }
        \]
        We can substitute $\rho(\infty|t,\lambda) = \sum_{d} Poisson(d;D(t,\lambda)) \rho(\infty|t,d) $ into RHS of \Cref{eq1_lemma_prob_on_giant_homo}, which along with the expression for forward distribution gives
        \[
        1- \rho(\infty|t,d) = (q_{t} ( 1- \sum_{d} \tilde{p}_{d}^{t}\rho(\infty|t,d)  ) + (1-q_{t}) (1- \sum_{d} \tilde{p}_{d}^{t'}\rho(\infty|t',d) ) )^{d}
        \]
        Suppose that $1-\rho(\infty|h,d) > 1- \rho(\infty|l,d)$ for some $d$, then from the expression above, it must also be the case for all $d$.
        In particular, it must be true for $d=1$. Also, let $ \frac{1-\rho(\infty|h,1)}{1- \rho(\infty|l,1)} = \alpha  $. Notice that $1-\rho(\infty|t,d) = (1-\rho(\infty|t,1))^{d} $, and we have $ \frac{1-\rho(\infty|h,d)}{1- \rho(\infty|l,d)} = \alpha^{d}  $. Therefore, we can write
\[
        \alpha (1- \rho(\infty|l,1) )= q_{h} (  \sum_{d} \tilde{p}_{d}^{h}\alpha^{d}(1-\rho(\infty|l,1))^{d}  ) + (1-q_{h})  \sum_{d} \tilde{p}_{d}^{l} (1- \rho(\infty|l,1))^{d}  
        \]
        \[
        1- \rho(\infty|l,1) = q_{l} (  \sum_{d} \tilde{p}_{d}^{l}(1- \rho(\infty|l,1) )^{d} ) + (1-q_{l}) (\sum_{d} \tilde{p}_{d}^{h}\alpha^{d} (1- \rho(\infty|l,1))^{d} ) 
        \]
        I now show that there can exist no $\rho(\infty|l,1)$ and $\alpha > 1$ where the two equalities hold at the same time. I separate this into two case. If $q_{h} < 1-q_{l}$, then $\alpha>1$ implies that RHS of the first equation is smaller than that of the second, which contradicts with $\alpha>1$.  

        If $q_{h} > 1-q_{l}$,
        for $\alpha$ to be greater than 1, we must have $\sum_{d} \tilde{p}_{d}^{h}\alpha^{d}(1-\rho(\infty|l,1))^{d} > \sum_{d} \tilde{p}_{d}^{l}(1- \rho(\infty|l,1) )^{d} $. Therefore, we must choose some $1 - \rho(\infty|l,1) > \sum_{d} \tilde{p}_{d}^{l}(1- \rho(\infty|l,1) )^{d}$, otherwise the second equation can't hold. 

Given such $1 - \rho(\infty|l,1)$, if $\alpha=1$, the LHS of both equations must be larger than the RHS. The largest value of $\alpha$ is when $\alpha (1- \rho(\infty|l,1))=1$. WHen $\alpha$ takes this value, the LHS of the first equation is 1 and RHS is less than 1. Also, notice that the LHS of the first equation is linear in $\alpha$ while the RHS is convex. These facts imply that LHS of the first equation must always be larger than the RHS for all $\alpha \ge 1$. This concludes the proof for \Cref{lemma:probability_on_giant_homophily}. The analogous argument shows that $\alpha$ cannot be smaller than 1 when degree distribution and thus forward distributions are the same for both types. This shows \Cref{lemma:probability_on_giant_baseline}.
	\qed

For any network signal strategy $\Pi^{n}=(\pi,Z)$ in the game with $n$ receivers, let $a^{*}(s,t,d;\Pi^{n})$ denote optimal action of receivers with type $t$ and degree $d$ when they observe $s$, for $s \in S \cup \emptyset$. Let $N_{t,d,s}(g^{n},z)$ denote the number of receivers of type $t$ and degree $d$ who observe signal $s $ given network $g^{n}$ and seeds $z$. 
	Define function $\hat{V}(d,\Pi^{n})$ to be
	\begin{multline} \label{hatvdefinition}
		\hat{V}(d,\Pi^{n}) \equiv \sum_{\omega} \mu_{s}(\omega) \sum_{s \in S \cup \{\emptyset\} }\pi(s|\omega)  \bigg[\sum_{g^{n}}Pr(g^{n}) \\
		v\big(\sum_{t}\gamma_{t}\frac{N_{t,d,s}(g^{n},z)}{N_{t,d}}a^{*}(s,t,d;\Pi^{n}) + \sum_{t}\gamma_{t}(1-\frac{N_{t,d,s}(g^{n},z)}{N_{t,d}})a^{*}(\emptyset,t,d;\Pi^{n})\big)    
		\bigg] 
	\end{multline}
	Using this function, the next lemma formally presents (part of) the imitation argument.
	\begin{lemma} \label{lemmaboundvhat}
		Suppose that $\zeta(d,h) \ge \zeta(d,l)$. If $v$ is convex or linear, for any $\epsilon >0$ and $d$, there exists some $n_{1}$ large enough such that for all $n \ge n_{1}$, we have $\hat{V}(d,\Pi^{n})  \le V_{pub}^{*} + \epsilon$ for all possible network signal strategy $\Pi^{n}$.   
	\end{lemma}
	\textit{Proof.}\\   
			For any $ \beta'>0$, we can find a $\beta>0$ such that $sup( v(x+2\beta)-v(x)) < \beta'$. For this $\beta$, by lemma \ref{lemma:observesignal} we can find some $n_{1}$ such that for all $n>n_{1}$, we have $Pr(g^{n} \in \mathbf{A}_{n}(\beta,\overline{d} )  ) > 1-\beta$. Let $S_{good}$ denote set of signals in $S$ that persuade both types of receivers, $S_{int}$ denote set of signals in $S$ that persuade only type $h$ receivers and $S_{bad}$ denote those that persuade neither group. 
			 For $j \in \{good,int,bad\}$, let \[\hat{\pi}_{j} (t,\omega,\Pi^{n}) \equiv  \sum_{s \in S_{j}} \pi(s|\omega)\sum_{g^{n}}Pr(g^{n})\frac{N_{t,d,s}(g^{n},z)}{N_{t,d}}\]
			
			We consider 2 cases. First, suppose that $a^{*}(\emptyset,l,d;\Pi^{n})=0$. We know that $a^{*}(s,l,d;\Pi^{n}) \ge a^{*}(\emptyset,l,d;\Pi^{n})$ for all $s \in S$. By definition of $\mathbf{A}_{n}(\beta,\overline{d})$, we know that $\frac{N_{l,d,s}(g^{n},z)}{N_{l,d}} - \frac{N_{h,d,s}(g^{n},z)}{N_{h,d}}$ is smaller than $2\beta$ if $g^{n} \in \mathbf{A}_{n}(\beta,\overline{d})$ and $s \in S_{good}$. Using these facts, how $n_{1}$, $\beta$ are chosen and convexity of $v$, $\hat{V}(d,\Pi^{n})$ is smaller than 
		\begin{multline} \label{lemma13eq3}
			\sum_{\omega} \mu_{s}(\omega)\bigg[
			\bigg( 1-\sum_{j \in \{good,int,bad\}} \hat{\pi}_{j} (h,\omega,\Pi^{n})                          \bigg)v\big(\gamma_{h}a^{*}(\emptyset,t,d;\Pi^{n})\big) \\
			+\hat{\pi}_{good} (h,\omega,\Pi^{n})v(1)      
			+\hat{\pi}_{int} (h,\omega,\Pi^{n})v(\gamma_{h})            
			\bigg] + v(1) \beta+\beta'
		\end{multline}
		
		Consider sender's public signal strategy $\pi_{2}$ described as follows. 
		For $j \in \{good,int,bad\}$
		\[
		\pi_{2}(s_{j}|\omega) = \hat{\pi}_{j} (h,\omega,\Pi^{n}) \qquad \pi_{2}(\emptyset|\omega) = 1 -\sum_{j \in \{good,int,bad\}} \hat{\pi}_{j} (h,\omega,\Pi^{n})  
		\]
		By lemma \ref{lemma:non_empty_posterior} and construction of $\pi_{2}$, all receivers choose action 1 upon observing signal $s_{good}$ with the public signal structure $\pi_{2}$. Similarly, type $h$ receivers choose action 1 upon observing $s_{int}$ and no receiver choose action 1 upon observing $s_{bad}$. 
		
		Also, receivers observe $\emptyset$ in state $\omega$ with the same probability as receivers with type $h$ and degree $d$ under network signal structure $\Pi^{n}$. Therefore, type $h$ receivers choose $\alpha(\emptyset,h,d;\Pi^{n})$ upon observing $\emptyset$. Therefore, sender's payoff using $\pi_{2}$ (as $n$ goes to infinity) is expression \ref{lemma13eq3} minus $\beta v(1) + \beta'$. Since sender's limiting payoff from $\pi_{2}$ is smaller than $V_{pub}^{*}$ and $\beta,\beta'$ can be arbitrarily small, the result follows for this case.
		
		In second case, suppose $\alpha(\emptyset,l,d;\Pi^{n})=1$. We know that $\alpha(s,h,d;\Pi^{n}) \le 1$ for all $s$. Using this fact and convexity of $v$, expression $\hat{V}(d,\Pi^{n})$ is smaller than 
		\begin{multline} \label{lemmaeq5}	
		 \beta v(1)+\beta' + \sum_{\omega} \mu_{s}(\omega)\bigg[
			\bigg( 1-\sum_{j \in \{good,int,bad\}} \hat{\pi}_{j} (l,\omega,\Pi^{n})                             \bigg)v\big(1\big) \\
			+\hat{\pi}_{good} (l,\omega,\Pi^{n})v(1)      
			+\hat{\pi}_{int} (l,\omega,\Pi^{n})v(\gamma_{h})           
			\bigg]  \equiv \hat{V}^{l}(d,\Pi^{n}) + \beta v(1)+\beta'
		\end{multline}
		Consider sender's public signal strategy $\pi_{3}$ described as follows.
		For $j \in \{good,int,bad\}$
		\[
		\pi_{3}(s_{j}|\omega) = \hat{\pi}_{j} (l,\omega,\Pi^{n})  \qquad \pi_{3}(\emptyset|\omega) = 1 -\sum_{j \in \{good,int,bad\}} \hat{\pi}_{j} (h,\omega,\Pi^{n})  
		\]
		By essentially the same argument as in the previous case, all receivers choose action 1 upon observing signal $s_{good}$ with the public signal structure $\pi_{3}$. Similarly, type $h$ receivers choose action 1 upon observing $s_{int}$ and no receiver choose action 1 upon observing $s_{bad}$. 
		
		Also, receivers observe $\emptyset$ in state $\omega$ with the same probability that a receiver with type $l$ and degree $d$ observes $\emptyset$ under network signal structure $\Pi^{n}$. Therefore, type $l$ receivers choose action 1 upon observing $\emptyset$, and type $h$ receivers must also choose action 1 upon observing $\emptyset$. Sender's payoff from using public signals $\pi_{3}$ (as $n$ goes to infinity) is $ \hat{V}^{l}(d,\Pi^{n}) $, which is expression \ref{lemmaeq5} minus $\beta v(1)+\beta'$. Since sender's limiting payoff from $\pi_{3}$ is smaller than $V_{pub}^{*}$, and $\beta,\beta'$ can be arbitrarily small, the result follows.
		
		The proof above immediately implies the following result.
		\begin{cor} \label{cor:vhatl}
			If $\alpha(\emptyset,l,d;\Pi^{n})=1$, then $V_{pub}^{*} \ge \hat{V}^{l}(d,\Pi^{n})$, where $\hat{V}^{l}(d,\Pi^{n})$ is defined in expression \ref{lemmaeq5}.
		\end{cor}  \qed
    
    \textit{Proof of \Cref{prop:main_prop}:}\\
    For any $\beta'>0$, we can find a $\beta>0$ such that $sup( v(x+8\beta)-v(x)) < \beta'$. For this $\beta$, we can find some $\overline{d}$ such that $\sum_{d>\overline{d}} (p_{d}^{h}+p_{d}^{l}	) <\beta $. Fix this $\overline{d}$. By \Cref{lemma:observesignal}, \Cref{lemmaboundvhat} and \Cref{rem:degree_dist}, we can find some $n_{1}$ such that for all $n>n_{1}$, we have $\hat{V}(d,\Pi^{n}) \le V_{pub}^{*} + \beta'$ for all $\Pi^{n}$ and the following conditions are satisfied with probabilty at least $(1-\beta)$: \\  $\sum_{d > \overline{d}}  \frac{N_{h,d}+N_{l,d}}{n}  < \beta  $,  $\frac{N_{t,d}}{n}  \in [\gamma_{t} p_{d}^{t}(1-\beta) ,  \gamma_{t} p_{d}^{t} (1+ \beta)]  $ and $g^{n}\in  \mathbf{A}_{n}(\beta,\overline{d}) $ for all $d \le \overline{d}$, $t \in \{h,l\}$.
			
	Let $S_{good}$, $S_{int}$ and $S_{bad}$ have the same meaning as in proof of lemma \ref{lemmaboundvhat}. 
		In the game with $n>n_{1}$ receivers, by condition 1, sender's payoff from any $\Pi^{n}_{1}$ is smaller than,
		\begin{multline} \label{prop3eq1}
		\beta v(1) +	\sum_{\omega} \mu_{s}(\omega) \sum_{s \in S \cup \{ \emptyset\}} \pi_{1}(s|\omega)  \bigg[  \sum_{g^{n}}Pr(g^{n}) \\
			v\big(2\beta + \sum_{t,d \le \overline{d} } p_{d}^{t}\gamma_{t}
			\frac{N_{t,d,s}(g^{n},z)}{N_{t,d}}a^{*}(s,t,d;\Pi^{n}_{1}) +  p_{d}^{t}\gamma_{t}(1-\frac{N_{t,d,s}(g^{n},z)}{N_{t,d}})a^{*}(\emptyset,t,d;\Pi^{n}_{1})   \big) 
			\bigg]  
		\end{multline}
  
		We consider two cases. In the first case, suppose that $a^{*}(\emptyset, l, d ;\Pi_{1}^{n}) =0 $ for all $d \le \overline{d}$. In this case, by definition, we also know that $a^{*}(s, l, d ;\Pi_{1}^{n}) =0 $ for any signal $s \notin S_{good}$. 
		Conditional on $g^{n} \in \mathbf{A}^{n}(\beta,\overline{d})$, we know that $\frac{N_{l,d_{1},s}(g^{n},z)}{N_{l,d_{1}}} - \frac{N_{l,d_{2},s}(g^{n},z)}{N_{l,d_{2}}} \le 2\beta$ for any $d_{1},d_{2} \le \overline{d}$ and $d_{1} \le d_{2}$. Since degree distribution for type $h$ FOSD that of type $l$, by how $\beta$ is chosen, condition 1 and convexity of $v$, \ref{prop3eq1} is smaller than 
		\begin{equation*} 
			\sum_{d \le \overline{d}}p_{d}^{h} \hat{V}^{n}(d,\Pi^{n}_{1}) + \beta  v(1) + \beta' \le V_{pub}^{*}+ \beta  v(1) +2 \beta'
		\end{equation*}
		The result for this case then follows since $\beta$ and $\beta'$ can be arbitrarily small. 
		
		Now consider the second case that there exists some $d$ such that $a^{*}(\emptyset,l,d;\Pi^{n}_{1}) = 1$. Let $d_{min}$ denote the smallest $d \le \overline{d}$ such that $a^{*}(\emptyset,l,d_{min};\Pi^{n}_{1})=1$ and $a^{*}(\emptyset,l,d;\Pi^{n}_{1})=0$ for all $d < d_{min}$. Note that this implies $a^{*}(\emptyset,l,d_{min};\Pi_{1}^{n} ) \ge a^{*}(\emptyset,t,d;\Pi_{1}^{n} )$ for all $d$, $t$. Again, if $g^{n} \in \mathbf{A}^{n}(\beta,\overline{d})$, we know that $\frac{N_{t,d_{1},s}(g^{n},z)}{N_{t,d_{1}}} - \frac{N_{t,d_{2},s}(g^{n},z)}{N_{t,d_{2}}} \le 2\beta$ for any $d_{1},d_{2} \le \overline{d}$ and $d_{1} \le d_{2}$. From these, we know that, for any $s$ and $g^{n} \in \mathbf{A}^{n}(\beta,\overline{d})$,
  \begin{multline}
      \sum_{d < d_{min} } p_{d}^{l}\gamma_{l}
			\frac{N_{l,d,s}(g^{n},z)}{N_{l,d}}a^{*}(s,l,d;\Pi^{n}_{1}) +  p_{d}^{l}\gamma_{l}(1-\frac{N_{l,d,s}(g^{n},z)}{N_{l,d}})a^{*}(\emptyset,l,d;\Pi^{n}_{1}) \le \\
2\beta + \sum_{d < d_{min} } p_{d}^{h}\gamma_{l}
			\frac{N_{l,d,s}(g^{n},z)}{N_{l,d}}a^{*}(s,l,d;\Pi^{n}_{1}) +  p_{d}^{h}\gamma_{l}(1-\frac{N_{l,d,s}(g^{n},z)}{N_{l,d}})a^{*}(\emptyset,l,d;\Pi^{n}_{1}) + \\
 (  \sum_{d < d_{min} } p_{d}^{l}\gamma_{l} -  \sum_{d < d_{min} } p_{d}^{h}\gamma_{l})
			(\frac{N_{l,d_{min},s}(g^{n},z)}{N_{l,d_{min}}}a^{*}(s,l,d_{min};\Pi^{n}_{1}) +  (1-\frac{N_{l,d_{min},s}(g^{n},z)}{N_{l,d_{min}}}))
   \end{multline}
   Also, for all  $d \ge d_{min}$, $t$, $s$ and $g^{n} \in \mathbf{A}^{n}(\beta,\overline{d})$, we have 
   \begin{multline}
       p_{d}^{t}\gamma_{t}
			\frac{N_{t,d,s}(g^{n},z)}{N_{t,d}}a^{*}(s,t,d;\Pi^{n}_{1}) +  p_{d}^{t}\gamma_{t}(1-\frac{N_{t,d,s}(g^{n},z)}{N_{t,d}})a^{*}(\emptyset,t,d;\Pi^{n}_{1}) \le \\
     p_{d}^{t}\gamma_{t}
			\frac{N_{l,d_{min},s}(g^{n},z)}{N_{l,d_{min}}}a^{*}(s,t,d;\Pi^{n}_{1}) +  p_{d}^{t}\gamma_{t}(1-\frac{N_{l,d_{min},s}(g^{n},z)}{N_{l,d_{min}}}) +2 \gamma_{t}p_{d}^{t}\beta
   \end{multline}

  By how $\beta$ is chosen, convexity of $v$ and the previous two inequalities, expression \ref{prop3eq1} is smaller than 
		\begin{equation*} 
			\beta v(1) + \beta'+	\sum_{d < d_{min} }p_{d}^{h}\hat{V}(d,\Pi_{1}^{n}) + 	\sum_{d_{min} \le d \le \overline{d}} p_{d}^{h} \hat{V}^{l}(d_{min},\Pi_{1}^{n})
		\end{equation*}
		Since $\hat{V}(d,\Pi^{n}_{1})  \le V_{pub}^{*} + \beta'$ and $\hat{V}^{l}(d_{min},\Pi_{1}^{n}) \le V_{pub}^{*}$ (by \Cref{cor:vhatl}), the first point follows since $\beta$ and $\beta'$ can be arbitrarily small.  \qed

\textit{Proof of \Cref{prop:sceptics_more_connected}}

        Firstly, note that expected degree of $h$ is $E_{h}(\lambda)(\gamma_{h}E_{h}(\lambda) + \gamma_{l}E_{l}(\lambda)q ) < \frac{1}{\overline{x}} $, so $E_{h}(\lambda) < \frac{1}{q\gamma_{l}E_{l}(\lambda) \overline{x} } $. This implies that given any $m>0$, the probability that a type $h$ node has connectivity smaller than $\frac{1}{q\gamma_{l}E_{l}(\lambda) \overline{x} }m$ is at least $\frac{m-1}{m}$. Select some $m$ large enough such that $1/m < \epsilon$. Also, we know that $E_{l}(\lambda) > (\overline{x}/\gamma_{l} - \frac{\gamma_{h}}{\gamma_{l}^{2}\overline{x}})^{0.5} $.

        Therefore, we can find some $\overline{x}$ large enough such that expected $\lambda$ is arbitrarily large for type $l$ nodes, and at least $(1-\epsilon)$ of type $h$ nodes have connectedness weakly less than $\frac{m}{q \gamma_{l}\overline{x}E_{l}(\lambda)}$. 
    We can look at the subnetwork formed among type $l$ nodes. As discussed above, this is equivalent to looking at a network where $\lambda$ is set to 0 for all type $h$ nodes. Therefore, by lemma 5.6 in \citet{bollobas2007phase}, $\rho(\infty|l,\lambda)$ is the maximal solution to the system of equations: 
    \[
    \rho(\infty|l,\lambda) = 1 - exp(- \gamma_{l}  \lambda \sum_{\lambda'}f_{l}(\lambda')\lambda' \rho(\infty|l,\lambda')  )
    \]
    Fix any $x \in (0,1)$, and let $\rho(\infty|l,\lambda)=x$ for all $\lambda$. We can make $E_{l}(\lambda)$ arbitrarily large, and as it goes to infinity, the right hand sides of all equations go to 1. Therefore, LHS minus RHS is negative for all equations. If we let $\rho(\infty|l,\lambda)=1$ for all $\lambda$, the LHS minus RHS is positive for all equations. Then, by fixed point theorem (Poincare-Miranda theorem), there exists a solution to the system of equations where $\rho(\infty|l,\lambda) \ge x$ for all $\lambda$. Since $x$ is arbitrary, the size of the type $l$ giant component can be arbitrarily close to $\gamma_{l}$ for $\overline{x}$ large enough. The proportion of type $h$ nodes with degree 0 is at least $ (1-\epsilon)( PO(0; \frac{\gamma_{h}}{(q\gamma_{l}E_{l}(\lambda) \overline{x})^{2} }m  +\frac{m}{\overline{x}}) $, which converges to $(1-\epsilon)$ as $\overline{x}$ goes to infinity for any given $m$. Therefore, we can always select $m$ and $\overline{x}$ large enough such that it's arbitrarily close to 1. 
        
		Consider the following network signal strategy to be used for all $n$: $ \pi(s|\omega=1)=1$, $ \pi(s|\omega=0)=\frac{\mu_{l}(\omega=1) }{ \mu_{l}(\omega=0)}$ and $\pi(\emptyset|\omega) = 1- \pi(s|\omega)$. Seed is always chosen to be on $L_{1}$, the giant component. Type $h$ nodes with degree $0$ always chooses action 1, and type $l$ receivers who observe signal $s$ chooses action 1. Sender's payoff is at least
  \[
  \mu_{s}(\omega=1)v( \frac{N_{l,L_{1}}}{n} + \frac{N_{h,d=0}}{n} ) + \mu_{s}(\omega=0)\big[ \frac{\mu_{l}(\omega=1) }{\mu_{l}(\omega=0) } v(\frac{N_{l,L_{1}}}{n} + \frac{N_{h,d=0}}{n}  )+  (1-\frac{\mu_{l}(\omega=1) }{\mu_{l}(\omega=0) }) v( \frac{N_{h, d=0}}{n}  ) \big]
  \]
  We have shown that, we can select $\overline{x}$ and $m$ large enough such that $\frac{N_{l,L_{1}}}{n}$ and $\frac{N_{h,d=0}}{n}$ can be arbitrarly close to $\gamma_{l}$ and $\gamma_{h}$ respectively when $n$ goes to infinity. The result follows.
\qed

    \textit{Proof of \Cref{rem:conditions}:}
    
Point 1 is obvious from the definition of $q_{t}$. For point 2, note that degree distribution for type $t$ is a mix of poisson distributions, where the weight attached to intensity $\lambda'(E_{t}(\lambda)\gamma_{t} + E_{t'}(\lambda)\gamma_{t'})$ is $f_{t}(\lambda')$ for type $t$ receivers. Therefore, if $f_{t}(\cdot)$ FOSD $f_{t'}(\cdot)$, then the distribution over intensities of POisson distributions for type $t$ FOSD that of type $t'$, which in turn implies first order dominance for degree distributions. 

For point 3, we use the result from \citet{sadler2020} that type $t$ is more connected than type $t'$ if $\frac{ \sum_{\lambda'}f_{t}(\lambda')PO(d; D(t,\lambda')   )  }{  \sum_{\lambda'}f_{t'}(\lambda')PO(d;  D(t',\lambda')   ) }$ is weakly increasing in $d$.
By conditions on $f_{t}$, the derivative of the ratio above has the same sign as
\[
\frac{\sum_{\lambda'}f_{t'}(\lambda') (\lambda')^{d} e^{ -r  D(t,\lambda') }ln(rD(t,\lambda'))   }{ \sum_{\lambda'}f_{t'}(\lambda') (\lambda')^{d} e^{ -r  D(t,\lambda') }   } - \frac{\sum_{\lambda'}f_{t'}(\lambda') (\lambda')^{d} e^{ -\lambda'  D(t',\lambda') }ln(D(t',\lambda'))   }{ \sum_{\lambda'}f_{t'}(\lambda') (\lambda')^{d} e^{ -\lambda'  D(t',\lambda') }   }
\]
Note that $ln(rD(t,\lambda)) = ln(\lambda r E_{t'}(\lambda) (\gamma_{t}r+q\gamma_{t'} ) )$, and $ln(D(t',\lambda)) = ln(\lambda  E_{t'}(\lambda) (\gamma_{t}qr+\gamma_{t'} ) )$. Let $\lambda_{min}$ and $\lambda_{max}$ denote the smallest and largest values that $f_{t'}(\cdot)$ attaches positive probabilities to. The difference above is then at least $ln(rD(t,\lambda_{min})) -ln(D(t',\lambda_{max})) $ for all $d$, which will become positive for $r$ large enough. 

For the special case of the island model, the ratio above is $\big(\frac{\gamma_{t} + q \gamma_{t'}  }{ \gamma_{t}q +\gamma_{t'}} \big)^{d}\frac{e^{-\lambda^{2}(\gamma_{t} + q \gamma_{t'} )} }{e^{-\lambda^{2}(\gamma_{t}q +\gamma_{t'}) } }$, which is increasing in $d$ if $\gamma_{t}>\gamma_{t'}$ and decreasing in $d$ if $\gamma_{t}<\gamma_{t'}$.
\qed

		\textit{Proof of \Cref{prop:risk_aversion}:}
        
Let the sender use the following network signal strategy for all $n$
\[
\pi_{1}(s'|\omega=1) = r (1 - \kappa) \qquad \pi_{1}(s'|\omega=0) = \frac{ \mu_{h}(\omega=1)}{ \mu_{h}(\omega=0) } r (1 - \kappa) \qquad \pi(\emptyset|\omega) = 1- \pi_{1}(s'|\omega)
\]
		where $r \ge 1$ and $\kappa=(\frac{\mu_{h}(\omega=1)}{\mu_{h}(\omega=0)}-1)/(\frac{\mu_{h}(\omega=1)}{\mu_{h}(\omega=0)}-\frac{\mu_{l}(\omega=1)}{\mu_{l}(\omega=0)})$. 
        Sender always chooses one seed on $\hat{L}_{1}$. It can be verified that type $h$ receivers must lways choose action 1. Now, consider the posterior of a type $l$ receiver with degree $d$ upon observing $\emptyset$. Using lemma \ref{lemma:observesignal}, we know that it converges to:
		\begin{equation} \label{p23}
			\frac{\mu_{l}(\omega=1)(1-r (1-\kappa) \hat{\zeta}(d,l))}{\mu_{l}(\omega=1)(1- \frac{ \mu_{h}(\omega=0)}{ \mu_{h}(\omega=0) } \alpha (1 - \kappa)\hat{\zeta}(d,l)) + \mu_{l}(\omega=1)(1- r (1-\kappa) \hat{\zeta}(d,l))}
		\end{equation}
        It can be verified that for any $r> 1$, expression \ref{p23} is strictly larger than 0.5 if $\hat{\zeta}(d,l)=1$ (see e.g. \citet{innocenti2022can} for a proof). We know that $\hat{\zeta}(d,l) $ goes to 1 as $d$ approaches infinity. Therefore, we can find some $\underline{d}$ large enouch such that for all larger $d$, their posterior convergse to some number larger than 0.5. Let the payoff function $v(.)$ be defined as follows, $v(x) = x$ for all $x \le \underline{x}$ and is equal to $\underline{x}$ for all larger $x$, where $\underline{x}> \gamma_{h}$ and $\underline{x} < \gamma_{h} + \gamma_{l}\sum_{d \ge \underline{d}} p_{d}^{l}(1 - \hat{\zeta}(d,l) ) $. For $n$ large enough, sender gets payoff $\underline{x}$ from the network strategy described above. With public signals, it's impossible to always persuade some type $l$ receivers, so sender's payoff must be less than $\underline{x}$.

If the sender has the constant relative risk aversion utility described in proposition, the sender's payoff from the strategy described above\footnote{Note that the sender gets slightly less payoff from $s'$ compared to $\emptyset$. Nonetheless, when $b$ gets large, the payoff function will be arbitrarily close to flat beyond $\overline{x}$ defined in the last paragraph. Therefore, we can find some function that arbitrarily close to the constant relative risk aversion payoff function, where the sender gets weakly lower payoff from $\emptyset$ than $s'$.} is at least $\frac{\overline{x}^{b}-1}{b}$, and from public signal is at most\\ $\mu_{s}(\omega=0) (1 -\frac{\mu_{l}(\omega=1)}{\mu_{l}(\omega=0)} )\frac{\gamma_{h}^{b}-1}{b}  $. The former will be larger for $b$ small enough.   \qed

\paragraph{Proof for Voting Game:}
	I first study the optimal public signals for the sender. If a signal doesn't persuade type $l$ receivers, then sender gets payoff 0 when that signal is realised. Therefore, sender's optimal strategy must maximise the probability that type $l$ receivers choose action 1. The problem is now equivalent to standard Bayesian persuasion with a single type $l$ agent. Since this is the same as standard Bayesian persuasion in \citet{kamenica2011bayesian}, the proof is omitted. Sender's optimal payoff using public signals (in the limit as $n \rightarrow \infty$) is $
    V_{pub}^{*}=\mu_{s}(\omega=1) + \mu_{s}(\omega=0)\frac{\mu_{l}(\omega=1)}{\mu_{l}(\omega=0)}$
	
	\textit{Proof of Proposition \ref{propvotinggame} :}\\
	By lemma \ref{lemma:probability_on_hatL}, it must be the case that $\hat{\zeta}(l,d)=0$ for all $d$, or $\hat{\zeta}(l,d)>0$ for all $d>0$. Also, if it's not the case that $\hat{\zeta}(l,d)=0$ for all $d$, then $\hat{\zeta}(l,d)$ converges to 1 as $d$ goes to infinity. This implies that as long as $\frac{\hat{L}_{1}}{n}$ doesn't converge in probability to 0, then $d_{vote}$ is finite. 
	One can verify that when $d_{vote}$ is finite, the ratio
	\[
	\frac{\mu_{l}(\omega=1)}{\mu_{l}(\omega=0)}\frac{1-  \frac{\mu_{h}(\omega=0)}{\mu_{h}(\omega=1)}   \hat{ \zeta }(l,d)  }{1-   \hat{ \zeta }(l,d) } 
	\]
	is weakly larger than 1 iff $d \ge d_{vote}$. By definition, the ratio is smaller than 1 for all $d$ if $d_{vote}$ is $\infty$.
	
	I show the necessary condition first. Suppose that the necessary condition stated in proposition doesn't hold, i.e. $\gamma_{h} + \sum_{d \ge d_{vote}}\gamma_{l}p_{d}^{l} < \overline{x}$. 
	By definition of $d_{vote}$ and hypothesis, we can find a $\beta>0$ small enough such that \[
		\hat{ \zeta }(l,d) +3 \beta <1	 \quad \text{ and } \quad  \frac{\mu_{l}(\omega=1)}{\mu_{l}(\omega=0)}\frac{1-  \frac{\mu_{h}(\omega=0)}{\mu_{h}(\omega=1)} (2\beta+  \hat{ \zeta }(l,d) ) }{1-   \hat{ \zeta }(l,d) - 3\beta} < 1
		\]
		for all $d < d_{vote}$.
	
	We know that the proportion of type $h$ receivers converge to $\gamma_{h}$, and $\sum_{d \ge d_{vote}}\frac{N_{l,d}}{n}$ converges in probability to $\sum_{d \ge d_{vote}} \gamma_{l} p_{d}^{l}$ (by \Cref{rem:degree_dist}). Therefore, for fixed $\beta$ satisfying conditions above and any $\overline{d}>d_{vote}$ (if $d_{vote}$ is infinity, then select $\overline{d}$ to be large enough such that $\gamma_{h} + \sum_{d > \overline{d} }\gamma_{l}p_{d}^{l} < \overline{x}$), we can find some $n_{1}$ such that for all $n>n_{1}$, with probability at least $1-\beta$, we have $g^{n} \in \mathbf{A}^{n}(\beta,\overline{d})$
	
	Suppose that sender uses some strategy $\Pi^{n}$ in the game with $n \ge n_{1}$ receivers. Let $S_{good}$ denote the set of signals that convince both types of receivers, $S_{int}$ denote the set of signals that convince only type $h$ receivers and $S_{bad}$ as those that convince neither. 
	Let $OB_{good}(d,\omega;\Pi^{n}) $ denote the probability that a type $l$ receiver with degree $d$ observe observe signals in $S_{good}$ under $\Pi^{n}$. Similarly let $OB_{int}(d,\omega;\Pi^{n}) $ denote the probability that a type $l$ receiver with degree $d$ observe observe signals in $S_{int}$ under $\Pi^{n}$. We know that the probability of observing signals in $S_{bad}$ is smaller than $\beta$. As in other proofs, let $a^{*}(\emptyset,l,d;\Pi^{n})$ denote actions of receivers with type $t$ and dgree $d$ upon observing $\emptyset$.
	
	I
    show $a^{*}(\emptyset,l,d;\Pi^{n})=0$ for all $d <min\{ d_{vote},\overline{d}\}$ by contradiction. Suppose that $\alpha(\emptyset,l,d;\Pi^{n})=1$ for some $d < min\{ d_{vote},\overline{d}\}$. The likelihood ratio for type $l$ receiver with degree $d$ upon observing $\emptyset$ must be weakly greater than 1, therefore we must have
	\[
	\frac{\mu_{l}(\omega=1)}{\mu_{l}(\omega=0)}\frac{1- OB_{good}(d,\omega=1;\Pi^{n})-OB_{int}(d,\omega=1;\Pi^{n}) }{1- OB_{good}(d,\omega=0;\Pi^{n})-OB_{int}(d,\omega=0;\Pi^{n})-\beta} \ge 1
	\]
	By \Cref{lemma:non_empty_posterior}, we know that $\frac{OB_{good}(d,\omega=1;\Pi^{n})}{OB_{good}(d,\omega=0;\Pi^{n})} \ge \frac{\mu_{l}(\omega=0)}{\mu_{l}(\omega=1)}$ and $\frac{OB_{int}(d,\omega=1;\Pi^{n})}{OB_{int}(d,\omega=0;\Pi^{n})} \ge \frac{\mu_{h}(\omega=0)}{\mu_{h}(\omega=1)}$. Therefore, we must have 
	\[
	\frac{\mu_{l}(\omega=1)}{\mu_{l}(\omega=0)} \frac{1-   \frac{\mu_{h}(\omega=0)}{\mu_{h}(\omega=1)}K }{1-  K -\beta} \ge 1
	\]
	for any $ K \in [OB_{int}(d,\omega=0;\Pi^{n}), 1-\beta) $.
	
	$OB_{int}(d,\omega=0;\Pi^{n})$ is at most the probability of observing $s$ conditional on it being realised, which is at most $2\beta + \hat{\zeta}(l,d)$ since $g^{n} \in \mathbf{A}^{n}(\beta,\overline{d})$ with probability at least $1-\beta$.
	By how $\beta$ is chosen, we have $2\beta + \hat{\zeta}(l,d) < 1-\beta$, which in turn implies 
	\[
	\frac{\mu_{l}(\omega=1)}{\mu_{l}(\omega=0)}\frac{1-  \frac{\mu_{h}(\omega=0)}{\mu_{h}(\omega=1)} (2\beta+  \hat{ \zeta }(l,d) ) }{1-   \hat{ \zeta }(l,d) - 3\beta} \ge 1
	\]
	But this inequality is a contradiction with how $\beta$ is chosen, so $a^{*}(\emptyset,l,d;\Pi^{n})=0$ for all $d <min\{ d_{vote},\overline{d}\}$.
    Therefore, an upper bound on sender's payoff when signal $s$ does not persuade type $l$ receivers is $v(\frac{1}{n} (N_{h}+ \sum_{d \ge min\{ d_{vote},\overline{d}\}} N_{l,d}   )) $
    by assuming that all type $h$ receivers take action 1, and all type $l$ receivers with degree higer than $min\{ d_{vote},\overline{d}\}$ choose action 1. By \Cref{rem:degree_dist}, this converges in probability to $ v(\gamma_{h} + \sum_{d \ge min\{ d_{vote},\overline{d}\}}\gamma_{l}p_{d}^{l})=0$, since $\gamma_{h} + \sum_{d \ge min\{ d_{vote},\overline{d}\}}\gamma_{l}p_{d}^{l} < \overline{x}$. Therefore, sender can only get 0 unless a signal that persuades both types is realised, which implies that $\limsup_{n} V^{*}_{net}(n)$ is weakly smaller than $
		\sum_{\omega}\mu_{s}(\omega) \sum_{s \in S_{good}}\pi(s|\omega) $ for $n$ large enough. Since $\sum_{s \in S_{good}}\pi(s|\omega=1) \le 1$ and  $\sum_{s \in S_{good}}\pi(s|\omega=0) \le \frac{\mu_{l}(\omega=1)}{\mu_{l}(\omega=0)}$, this upper bound is weakly smaller than $V_{pub}^{*}$ found above.
	This concludes proof for the necessary condition. 
	
	To show the sufficient condition, suppose that sender uses strategy $\Pi_{2}^{n}$ below for all $n$,
	\[
	\pi_{2}(s|\omega=1)= \frac{\mu_{h}(\omega=0)}{\mu_{h}(\omega=1)} \qquad \pi_{2}(s|\omega=0)=1 \qquad \pi_{2}(\emptyset|\omega) = 1-\pi_{2}(s|\omega)
	\]
	Let the seeds always hit $\hat{L}_{1}$. One can verify that type $h$ receivers always choose action 1, regardless of what they observe. Upon observing $\emptyset$, receivers with type $l$ and degree $d$ has likelihood ratio
	\[
	\frac{\mu_{l}(\omega=1)}{\mu_{l}(\omega=0)} \frac{1-\pi_{2}(s|\omega=1)E(\frac{N_{\hat{L}_{1},l,d}}{N_{l,d}}) }{1-\pi_{2}(s|\omega=0)E(\frac{N_{\hat{L}_{1},l,d}}{N_{l,d}})}
	\]
    By lemma \ref{lemmaprobongiant}, $\pi_{2}(s|\omega)E(\frac{N_{\hat{L}_{1},l,d}}{N_{l,d}})$ converges to $\pi_{2}(s|\omega)\hat{ \zeta }(l,d)$, and thus the likelihood ratio converges to 
	$
	\frac{\mu_{l}(\omega=1)}{\mu_{l}(\omega=0)} \frac{1- \frac{\mu_{h}(\omega=0)}{\mu_{h}(\omega=1)}\hat{ \zeta }(l,d) }{1-\hat{ \zeta }(l,d)}
	$. This implies that for any $d > d_{vote}$, type $l$ receivers with degree $d$ choose action 1 upon observing $\emptyset$ for $n$ large enough. Therefore, for any $\overline{d}$, we can find some $n_{2}$ such that for all larger $n$, type $l$ receivers with degrees between $d_{vote}$ and $\overline{d}$ choose action 1 upon observing $\emptyset$. 

	If the sufficient condition holds, we can find some $\overline{d}$ such that $\gamma_{h} + \sum_{d=d_{vote}+1}^{\overline{d}} \gamma_{l}p_{d}^{l} (1-\hat{ \zeta }(l,d)) > \overline{x}$. By the argument above, for $n \ge n_{2}$, the proportion of receivers choosing action 1 in any realization is at least $\frac{1}{n} (N_{h} + \sum_{d=d_{vote}+1}^{\overline{d}}(N_{l,d} - N_{l,d,\hat{L}_{1}}(g^{n}))   )$.
    By \Cref{rem:degree_dist} and \Cref{lemmaprobongiant}, it
	converges in probability to $\gamma_{h} + \sum_{d=d_{vote}+1}^{\overline{d}} \gamma_{l}p_{d}^{l}(1-\hat{ \zeta }(l,d)) > \overline{x}$, so sender's payoff
	converges in probability to $v(1)$. This is strictly larger than $V_{pub}^{*} = \mu_{s}(\omega=1)v(1) + \mu_{s}(\omega=0)\frac{\mu_{l}(\omega=1)}{\mu_{l}(\omega=0)}v(1)$. 
	\qed
    	\subsection{Proofs for technical lemmas in section \ref{sec:analysis}}
        

    \textit{Proof of \Cref{lemma:non_empty_posterior}:}
    
		Given type $t$ and degree $d$, the (subjective) probability of observing signal $s\ne \emptyset$ and being in state $\omega$ is
		\[\mu_{t}(\omega)\pi(s|\omega) Pr(\text{observing }s|s \text{ is realised},\omega,d,t)\]
        Conditional on signal $s$ and a realised network, the set of nodes being selected as seeds and thus receivers observing the signal is completely determined, and therefore does not depend on the state $\omega$. Therefore, we know that $Pr(\text{observing }s|s \text{ is realised},\omega,d,t,g^{n})$
        is independent of $\omega$. In addition, we know that conditional on $s$, the network formation process is independent of state. This implies that $Pr(g^{n}|s\text{ is realised},\omega,d,t)$ is independent of $\omega$. Therefore, 
        \[
        Pr(\text{observing }s|s \text{ is realised},\omega,d,t) = Pr(\text{observing }s| s\text{ is realised},d,t)
        \]
        The rest then follows from Beyesian updating. \qed
        
\textit{Proof of \Cref{rem:existence_of_eq}:} It directly follows from \Cref{lemma:non_empty_posterior} and the discussion after it. \Cref{lemma:non_empty_posterior} pins down people's behaviour for each $s \in S$, and thus the group who observe $s$ given network and seeds. These in turn pin down the probabilities that receivers observe each $s$, and thus also $\emptyset$ which is observed with complementary probability. From this, we can get receivers' behaviour upon observing $\emptyset$. Given a signal and network, all these together determine each receiver's observation and their behaviour, and thus the sender's payoff. \qed

To prove the network related lemmas, I need to introduce some extra notation. 
 I define the following Branching processes. Let $\mathbf{B}$ denote the following branching process:
\begin{enumerate}
    \item The process starts with a single node in generation 1. The node has type $t$ with probability $\gamma_{t}$, and conditional on $t$ has connectivity $\lambda$ with probability $f_{t}(\lambda)$.
    \item Given type $t$ and connectivity $\lambda$, all nodes in this process produce nodes according to Poisson distribution with intensity $D(t,\lambda)$. Recall that $D(t,\lambda) \equiv \lambda (\gamma_{t}E_{t}(\lambda) +q \gamma_{t'} E_{t'}(\lambda)) $.
    \item Conditional on type $t$, connectivity $\lambda$ and having $d$ offsprings, the distribution over types and connectivity of these offsprings is given by a multinomial distribution with $d$ trials, and success probability for pair $(t',\lambda')$ is $\frac{ \lambda \lambda' f_{t'}(\lambda')\gamma_{t'} }{D(t,\lambda)}$ if $t=t'$, and $\frac{ \lambda \lambda' f_{t'}(\lambda')\gamma_{t'} q }{D(t,\lambda)}$ if $t \ne t'$. 
    \item All nodes produce offsprings mutually independently of each other. 
\end{enumerate}

Also, let $\rho(<m)$ and $\rho(m)$ denote the probability that the process contains less than $m$ and exactly $m$ nodes. Let $\rho(\infty)$ denote the probability that the process never dies out. One can show that $\rho(\infty) = 1- \lim_{m \rightarrow \infty}\rho(<m)$, see e.g. \citet{bollobas2007phase}. Similarly, let $\rho(<m|\cdot)$, $\rho(m|\cdot)$ and $\rho(\infty|\cdot)$ denote the respective probabilities conditional on properties of the initial node.

	Let $N_{<m}(g^{n})$ denote the number of nodes on components of size strictly less than $m$ in the realised network $g^{n}$. Let $N_{L_{1}}(g^{n})$ denote the number of nodes on $L_{1}$. 

I'll prove lemmas \Cref{lemma:only_observe_viral} and \Cref{lemma:signal_need_to_be_viral} together in the following proof. 

		\textit{Proof of \Cref{lemma:only_observe_viral} and \Cref{lemma:signal_need_to_be_viral}} 
		
		An upper bound on the proportion of receivers outside $L_{1}$ that observe $s$ can be obtained by assuming that all receivers on components larger than size $\overline{m}$ observes $s$, and all seeds are on separate components of size $\overline{m}$. Summing these up, the proportion of receivers not on $L_{1}$ who can observe $s$ is at most \[
		  \frac{ n- N_{<\overline{m}}(g^{n}) - N_{L_{1}}(g^{n})}{n} + \frac{n^{\alpha}\overline{m}}{n}\]
		We know that $N_{<\overline{m}}(g^{n})/n $ converges in probability to (by Theorem 9.1 from \citet{bollobas2007phase}) $\rho(<\overline{m})$, and $\frac{ N_{L_{1}}}{n}$ converges in probability to $\rho(\infty)$. For any $\epsilon>0$, We can select $\overline{m}$ large enough such that $1-\rho(<\overline{m}) - \rho(\infty)$ is smaller than $\epsilon$. For any given $\overline{m}$, $\frac{n^{\alpha}\overline{m}}{n}$ converges to 0. The first parts of \Cref{lemma:only_observe_viral} and \Cref{lemma:signal_need_to_be_viral} then follow since $\epsilon$ can be arbitrary. 
		
		Next, we consider signals that persuade only type $h$ receivers and thus only passed on by them. 
		We first bound the proportion of type $h$ nodes not on $\hat{L}_{1}$ that can observe this signal. By the argument preceeding \Cref{cor:giant_in_h_subnetwork}, to study the type $h$ subnetwork, we can essentially study the transformed network formation process where $\lambda_{i}=0$ for all type $l$ receivers. Then, applying the same steps from last paragraph shows the proportion of type $h$ nodes that are not on $\hat{L}_{1}$ and observe $s$ converges in probability to 0, if $s$ persuades only type $h$ receivers. 
		
		A type $l$ node can receive the signal only if they have some type $h$ neighbour who receives the signal. Therefore, we proceed by bounding the number of neighbours of type $h$ nodes who receive signal and are not on $\hat{L}_{1}$. Since we know that $\sum_{d} p_{d}^{h}d$ is bounded, we can find some $\overline{d} >0$ such that $\sum_{d \ge \overline{d}} p_{d}^{h}d < \epsilon$.

        The proportion of type $h$ receivers who are not on $\hat{L}_{1}$ and observe $s$ converges in probability to 0, so the proportion of type $h$ receivers that in addition has less than $\overline{d}$ neighbors must also converge to 0. This in turn implies that the number of neighbours they have converge in probability to 0. Therefore, we can provide the upper bound by assuming that all type $h$ nodes with degree higher than $\overline{d}$ observe the signal and are not on $\hat{L}_{1}$\footnote{For a signal $s$ that only persuades type $h$ receivers, if no seeds are on $\hat{L}_{1}$ or a type $l$ node has no neighbour on $\hat{L}_{1}$, then the type $l$ node can only get the signal $s$ from type $h$ receivers who are not on $\hat{L}_{1}$}. The number of neighbors they have is $\sum_{d \ge \overline{d}} N_{h,d} d $, and the ratio $\sum_{d \ge \overline{d}} \frac{N_{h,d} d }{n}$ converges in probability to $\sum_{d \ge \overline{d}} p_{d}^{h}d < \epsilon$. Since $\epsilon$ is arbitrary, the second parts of \Cref{lemma:only_observe_viral} and \Cref{lemma:signal_need_to_be_viral} follow. 
	\qed

To prove lemma \ref{lemmaprobongiant}, I first prove the following lemma, which says that if we look at local properties, the random network is approximated by the branching process above. Let $N_{t,d,\lambda,m}(g^{n})$ denote the number of nodes with type $t$, connectedness $\lambda$, degree $d$ and is on a component of size $m$ in network $g^{n}$.
\begin{lemma} \label{lemma:local_property_convergence}
    $\frac{N_{t,d,\lambda,m}(G^{n}) }{ n }$ converges in probability to $\gamma_{t}f_{t}(\lambda)Poisson(d;D(t,\lambda)) \rho(m|t,d,\lambda)$. 
\end{lemma}
The proof of this largely follows from \citet{sadler2020}'s proof for Theorem 10. 
Let $g_{i,r}^{n}$ denote the rooted graph formed by deleting all nodes that are more than $r$ steps away from $i$, and letting $i$ be the root. Note that given any realised branching process, we can connect each node to their offsprings and parent. Then, we can form a rooted graph and delete all nodes that are distance $r$ away from the initial node, and let the initial node be the root. 

 For these proofs, I introduce the following notation. Let $\eta(t,t',\lambda,\lambda') $ be $\lambda \lambda'f_{t'}(\lambda') \gamma_{t'}$ if $t=t'$ and $q \lambda \lambda' f_{t'}(\lambda') \gamma_{t'}$ if $t \ne t'$. This $\eta(\cdot)$ function captures the expected number of a node's neighbours of type $t'$ and connectedness $\lambda'$, given the node's type $t$ and connectedness $\lambda$. 
\begin{lemma} \label{lemma:local_property}
Let $\mathbf{P}^{r}$ be a finite set of rooted graph where all nodes are within distance $r$ of the root. Uniformly randomly select some node $i$, the probability that $g_{i,r}^{n} \in \mathbf{P}^{r}$ converges to the probability that branching process $B$ produces a rooted graph in $\mathbf{P}^{r}$. 
\end{lemma}
\begin{proof}
    It's sufficient to show that the claim is true for all $\mathbf{P}^{r}$ that's a singleton, i.e. it contains only one rooted graph.

    I do this by induction. Suppose that the claim above is true for $r$, I show that it must be true for $r+1$. Let $g_{r+1}$ be any rooted graph where all nodes are within distance $r+1$ of the root, and there's at least 1 node with distance $r+1$ from the root. Let $g_{r}$ denote the graph after deleting nodes with distance $r+1$ from the root node. Given the induction hypothesis, it's sufficient to show that $Pr(g_{i,r+1}^{n}=g_{r+1} |g_{i,r}^{n}=g_{r} ) $ converges to the probability that the branching process produces $g_{r+1}$ conditional on producing $g_{r}$. 

    Let $I$ be the set of nodes with distance $r$ from the root node in $g_{r}$. Let $d_{j,t',\lambda'}$ denote the number of neighbours of node $j$ in $g_{r+1}$ that have type $t'$, connectedness $\lambda'$ and are at distance $r+1$ from root. Note that by definition, $g_{r+1}$ is fully described by $g_{r}$ and $\{(d_{j,t',\lambda'})_{t',\lambda'}  | j \in I\}$. 
    
    First, we look at the probability that the branching process produces $\{(d_{j,t',\lambda'})_{t',\lambda'}  | j \in I\}$ conditional on $g_{r}$. Conditional on type $t_{j}$ and connectedness $\lambda_{j}$ of any node $j \in I$, the probability that it produces $(d_{j,t',\lambda'})_{t',\lambda'}$ is $\Pi_{t',\lambda'}PO(d_{j,t',\lambda'} ; \eta(t_{j},t',\lambda_{j},\lambda')  )   $, and it's independent of $g_{r}$ and how other nodes in $I$ produce nodes. Therefore, the probability that the branching process produces $\{(d_{j,t',\lambda'})_{t',\lambda'}  | j \in I\}$, and thus $g_{r+1}$, conditional on $g_{r}$ is
    \begin{equation} \label{equation: coupling_lemma_offspring_production}
        \Pi_{j \in I} \Pi_{t',\lambda'}PO(d_{j,t',\lambda'} ; \eta(t_{j},t',\lambda_{j},\lambda')  ) 
    \end{equation}
 
    Now we consider the network. Conditional on types and connectedness of all nodes, connection between any pair is independent of connections between other pairs. Let $N_{t',\lambda'}$ denote the number of nodes not in $I$ that have type $t'$ and connectedness $\lambda'$. Also, note that neighbours of nodes in $I$ that are at distance $r+1$ from root must be nodes not in $g_{r}$. Therefore, conditional on types and connectedness of all nodes and $g_{r}$, the probability of producing $\{(d_{j,t',\lambda'})_{t',\lambda'}  | j \in I\}$ is 
    \[
    \Pi_{j \in I} \Pi_{t',\lambda'} Binom(d_{j,t',\lambda'}|N_{t',\lambda'}, \frac{ \mathbf{1}(t_{j}=t') + \mathbf{1}(t_{j} \ne t')q   }{ n}\lambda_{j} \lambda' )
    \]
    Since $N_{t',\lambda'} \frac{ \mathbf{1}(t_{j}=t') + \mathbf{1}(t_{j} \ne t')q   }{ n}\lambda_{j} \lambda'$ converges in probability to $\eta(t_{j},t',\lambda_{j},\lambda')$, so the expression above converges in probability to expression \ref{equation: coupling_lemma_offspring_production}. 
    
    To complete the induction argument, we need to show that the claim holds true for $r=0$. For $r=0$, $g_{r}$ consists only of the root node, and thus is fully described by the type and connectedness of the node. The probability that the branching process produces a node with type $t$ and connectedness $\lambda$ is $\gamma_{t}f_{t}(\lambda)$. From our network formation process, as $n$ gets large, the probability that a randomly selected node has type $t$ and connectedness $\lambda$ has the same expression. 
\end{proof}

This shows that the expected proportion of nodes with any given local propoerty converges to the probabilty that $\mathbf{B}$ has this property. We wish to show that they converge in probability. Note that the proportion of nodes with type $t$ and connectedness $\lambda$ converges in probability to $\gamma_{t}f_{t}(\lambda)$. By remark 8.8 and also lemma A.5 in \citet{bollobas2007phase}, we can assume that connectedness for receivers are deterministic, and the proportion of nodes with type $t$ and connectedness $\lambda$ converges to $\gamma_{t}f_{t}(\lambda)$. This will not change the asymptotic distribution over networks. I will use this assumption for the rest of the proof for \Cref{lemma:local_property_convergence}

For any 2 realised networks $g_{1}^{n}$ and $g_{2}^{n}$ of the same size, I say that they are related by a deletion if one of them can be obtained by deleting one link from the other. Let $F(\cdot)$ be a real-valued function defined on the possible networks . We say that $F(\cdot)$ is $\beta$-Lipschitz if for any two networks related by a deletion, we have that $|F(g_{1}^{n})-F(g_{2}^{n}) | < \beta$. 
\begin{lemma} \label{lemma:lipschitz_function_converge}
    If $F$ is a $\beta$-Lipschitz function, then $F(g^{n})$ converges in probability to its expected value. Formally, for any $\epsilon>0$, 
    \[
    Pr(| F(g^{n}) - E(F(g^{n}))| \ge \epsilon ) \le 2 e^{-\frac{\epsilon^{2}}{2\beta^{2}n}}
    \]
\end{lemma}
\begin{proof}

    We consider a random process in which we seqentially reveal whether two nodes are connected. For each agent $i$, we sequentially reveal whether $i$ is connected to nodes with larger index. Note that starting from $i=1$ and going to $i=n$ reveals the whole network. 

    Let $M_{j}$ denote the realised connection patterns generated by the revelation process upto step $j$. The process $X_{j} = E(F(g^{n})|M_{j})$ is a martingale\footnote{To see this, note that it's sufficient to show that $E( E(F(g^{n})  |M_{j+1}|M_{j} ) = E(F(g^{n})  |M_{j})$, which holds almost by definition}. Also, we know that $| E( F(g^{n})|M_{j} ) -E( F(g^{n})|M_{j+1}  )| < \beta$. This holds since revelation of connections for steps later than $j+1$ are independent of realizations before. Therefore, given any $M_{j}$, the distributions over connections revealed after $j+1$-th step is the same as conditioning on $M_{j+1}$. By construction of the process, $M_{j}$ and $M_{j+1}$ can only differ at most by deletion of one link. This inequality along with Azuma's inequality produces the result.   
\end{proof}
To show \Cref{lemma:local_property_convergence}, we need to show that $N_{t,d,\lambda,m}(g^{n})$ are Lipschitz, which is shown in the following lemma. 
\begin{lemma} \label{lemma:counting_is_liptiz}
    Let $N_{t,d,\lambda,m}(g^{n})$  is $2m+2$-Lipschitz.
\end{lemma}
\begin{proof}
    Consider any network $g^{n}$ and delete some link. Note that if removing the link doesn't change the component structure. Then it can only change the degree of two agents and therefore will only change the number of nodes with the described property by 2.

    If it changes the component structure. It can only break some component into 2. If that component that's broken is of size $m$, then it will change the number of nodes with the property described by at most $m+2$. If it's of size more than $m$, it's possible that it's broken into two compoenents of size $m$. Therefore, it can change the number of nodes with the desired property by at most $2m+2$. 
\end{proof}

Combining \Cref{lemma:local_property}, \Cref{lemma:lipschitz_function_converge} and \Cref{lemma:counting_is_liptiz} gives \Cref{lemma:local_property_convergence}. Using \Cref{lemma:local_property_convergence}, we proof \Cref{lemmaprobongiant} in the following two lemmas, also using the following remark. 

\begin{rem} \label{rem:what_is_rho}
$\rho(\infty|t, d,\lambda)$ is given by
    \[1- ( \frac{\sum_{\lambda'} (1-\rho(\infty|t,\lambda')) \gamma_{t}f_{t}(\lambda')\lambda' +(1-\rho(\infty|t',\lambda')) q \gamma_{t'}f_{t'}(\lambda')\lambda' }{\gamma_{t}E_{t}(\lambda) + q\gamma_{t'}E_{t'}(\lambda) }  )^{d}\] Note that this expression does not depend on $\lambda$, so $\rho(\infty|t,d)$ is equal to the same expression. 
\end{rem}
\begin{proof}
This is standard in the study of branching processes (see. e.g. sections 5 and 6 in \citet{bollobas2007phase}) and I provide a brief argument here. The process never dies out if and only if at least one process that continues from one of the second generation offsprings does not die out. 

Conditional on the initial node having $(t,\lambda)$, each second generation node has type $(t,\lambda')$ with probability $\frac{  \gamma_{t}f_{t}(\lambda')\lambda'}{\gamma_{t}E_{t}(\lambda) + q\gamma_{t'}E_{t'}(\lambda)  }  $ and $(t',\lambda')$ with probability $\frac{\gamma_{t'}f_{t'}(\lambda')q\lambda' }{\gamma_{t}E_{t}(\lambda) + q\gamma_{t'}E_{t'}(\lambda)  }  $, independent of types and connectedness of other nodes.
    From construction, we know that if a second generation node has type $(t,\lambda)$, then the process that continues from it behaves in the same way as $\mathbf{B}$ with an initial node of type $t$ and connectedness $\lambda$. Therefore, each second generation node with type $t$ and connectedness $\lambda$ has a probability of dying out $1- \rho(\infty|t,\lambda)$. Since all nodes produce offsprings mutually independently from each other, conditional on types and connectedness of the second generation nodes, the processes that continue from them are also mutually independent.  Then, the expression follows. 
\end{proof}

 	\begin{lemma} \label{lemma:probability_on_L} 
		Among receivers having degree $d$, type $t$ and connectedness $\lambda$ with $f_{t}(\lambda)>0$, the proportion of those on the $L_{1}$ converges in probability to $\rho(\infty|t,d)$, which is increasing in $d$.
	\end{lemma}
	\begin{proof} 
		Let $N_{t,d,\lambda,\le m}(g^{n})$ be the number of nodes that have type $t$, degree $d$, connectedness $\lambda$ and is on a component of size smaller than $m$. Let $N_{t,d,\lambda}(g^{n})$ be the number of nodes that have type $t$, degree $d$ and connectedness $\lambda$. 
		
		By \Cref{lemma:local_property_convergence}, $N_{t,d,\lambda,\le m}(g^{n})/N_{t,d,\lambda}(g^{n})$ converges in probability to $ \rho_{D}(\le m|t,d,\lambda)$. Let $N_{t,d,\lambda,L_{1}}(g^{n})$ denote the number of nodes that have type $t$, degree $d$, connectedness $\lambda$ and are on the largest component $L_{1}$. 

By Theorem 9.1 from \citet{bollobas2007phase}, $\frac{N_{>m}(g^{n}) - N_{L_{1}}(g^{n})  }{N_{t,d,\lambda}}$ converges to in probability to $\frac{1- \rho(\le m) - \rho(\infty)}{\gamma_{t}f_{t}(\lambda)PO(d;D(t,\lambda)) }$. For any $\epsilon>0$, we can find some $\overline{m}$ large enough such that  $\frac{1- \rho(\le m) - \rho(\infty)}{\gamma_{t}f_{t}(\lambda)PO(d;D(t,\lambda)) }<\epsilon$ for all $m \ge \overline{m}$. We can also fine $\overline{m}$ large enough such that $ 1- \rho_{D}(\le \overline{m}|t,\lambda,d) \le \rho_{D}(\infty|t,\lambda,d) + \epsilon$.
		By definition, we know that 
		\[
		\frac{N_{t,d,\lambda,L_{1}}(g^{n})}{N_{t,d,\lambda}} + \frac{N_{t,d,\lambda,\le \overline{m}}(g^{n})}{N_{t,d,\lambda}} +  \frac{ - N_{t,d,\lambda,L_{1}}(g^{n}) + N_{t,d,\lambda,> \overline{m}}(g^{n})  }{N_{t,d,\lambda} }                               =1 
		\]
		The probability that $\frac{N_{t,d,\lambda,\le \overline{m}}(g^{n})}{N_{t,d,\lambda}}$ is smaller than $\rho_{D}(\le \overline{m}|t,d,\lambda) - \epsilon$ converges to 0. Therefore, the probability that 	
		\[ \frac{N_{t,d,\lambda,L_{1}}(g^{n})}{N_{t,d,\lambda}} \ge \rho_{D}(\infty|t,d,\lambda) + 2\epsilon \ge 1- \rho_{D}(\le \overline{m}|t,d,\lambda) + \epsilon  \] converges to 0. Also, since the probabilities that $\frac{N_{t,d,\lambda,\le \overline{m}}(g^{n})}{N_{t,d,\lambda}}>\rho_{D}(\le \overline{m}|t,d,\lambda) + \epsilon$ and $\frac{ - N_{t,d,\lambda,L_{1}}(g^{n}) + N_{t,d,\lambda,> \overline{m}}(g^{n})  }{N_{t,d,\lambda} } > \epsilon $ converge to 0.  Therefore, the probability that 
		\[
		\frac{N_{t,d,\lambda,L_{1}}(g^{n})}{N_{t,d,\lambda}} \le 1- \rho_{D}(\le \overline{m}|t,d,\lambda) - 2\epsilon 
		\]
		converges to 0. Therefore, the probability that $\frac{N_{t,d,\lambda,L_{1}}(g^{n})}{N_{t,d,\lambda}} \le \rho_{D}(\infty|t,d,\lambda) - 2\epsilon $ converges to 0. Since $\epsilon$ is arbitrary, this shows that $\frac{N_{t,d,\lambda,L_{1}}(g^{n})}{N_{t,d,\lambda}}$ converges in probability to $\rho_{D}(\infty|t,d,\lambda)$. This also implies that $\frac{N_{t,d,L_{1}}(g^{n})}{N_{t,d}}$ converges in probability to $\rho_{D}(\infty|t,d)$. The fact that $\rho_{D}(\infty|t,d)$ is increasing in $d$ then follows from remark \ref{rem:what_is_rho}. 
	\end{proof}

	\begin{lemma} \label{lemma:probability_on_hatL} 
		Among receivers having degree $d$ and type $t$, the proportion of those with a neighbour on the $ \hat{L}_{1}$ converges in probability, and the limit is increasing in $d$
	\end{lemma}
	\begin{proof}
    Using the argument preceeding \Cref{cor:giant_in_h_subnetwork} again, the distribution over subnetwork among type $h$ nodes is the same as the distribution over networks after setting $\lambda_{i}=0$ for all type $l$ nodes. I use $\hat{\rho}(\infty)$ to denote the probability that the associated process never dies out, and similarly $\hat{\rho}(\infty|\cdot)$ to denote the conditional probabilities. 

Let $N_{h,d_{h},\lambda}$ denote the number of nodes with type $h$, $d_{h}$ type $h$ neighbours and connectedness $\lambda$, and let $N_{h,d_{h},\lambda,\hat{L}_{1}}$ denote those that in addition are on $\hat{L}_{1}$. 
     Using \Cref{lemma:probability_on_L}, we get that $\frac{N_{h,d_{h},\lambda,\hat{L}_{1}} (g^{n})}{N_{h,\lambda,d_{h} }}$ converges in probability to $\hat{\rho}(\infty|h,d_{h},\lambda)$.

    A node of type $h$ and connectedness $\lambda$ is connected to some randomly selected type $l$ node with probability $\frac{\lambda E_{l}(\lambda) q  }{n}$. Also, whether it's connected is independent of all the other connections. This implies that the probability that it has $d_{l}$ type $l$ neighbours is $Binom(d_{l};N_{l}, \frac{\lambda \sum_{\lambda'} \lambda' f_{l}(\lambda') q  }{n})$, where $N_{l}$ is the number of type $l$ receivers. 
    Therefore, the number of nodes with type $h$, connectedness $\lambda$, $d_{h}$ high type neighbours, $d_{l}$ type $l$ neighbours and on $\hat{L}_{1}$ can be seen as the random variable formed in the following way: We first form the type $h$ subnetwork as described above, and then take $N_{h,d_{h},\lambda,\hat{L}_{1}}$ independent draws of the binomial random variable with $N_{l}$ trials and success probability $\frac{\lambda E_{l}(\lambda) q  }{n}$. 

    We know that $Binom(d_{l};N_{l}, \frac{\lambda E_{l}(\lambda) q  }{n})$ converges to $PO(d_{l};\gamma_{l} \lambda E_{l}(\lambda) q  )$, and $N_{h,d_{h},\lambda,\hat{L}_{1}}$ goes to infinity as $n$ goes to infinity, therefore
    \[
    \frac{N_{h,d_{h},d_{l},\lambda,\hat{L}_{1}}}{n} \rightarrow_{p}  \gamma_{h}f_{h}(\lambda)PO(d_{h} ; \lambda \gamma_{h}E_{h}(\lambda) )\hat{\rho}(\infty |d_{h},h,\lambda) PO(d_{l} , \lambda q 
    \gamma_{l}E_{l}(\lambda) )
    \]
    Similarly, we can use this logic to get $N_{h,\lambda,d_{h},d_{l} }$, i.e., We first form the type $h$ subnetwork as described above, and then take $N_{h,d_{h},\lambda}$ independent draws of the binomial random variable with $N_{l}$ trials and success probability $\frac{\lambda E_{l}(\lambda) q  }{n}$. We have
     \[
    \frac{N_{h,\lambda,d_{h},d_{l} } }{n} \rightarrow_{p} \gamma_{h}f_{h}(\lambda)PO(d_{h} ; \lambda  \gamma_{h}E_{h}(\lambda) ) PO(d_{l} , \lambda q \gamma_{l}E_{l}(\lambda) )
    \]
    From these, we have 
    $\frac{N_{h,\lambda,d_{h},d_{l},\hat{L}_{1}}}{N_{h,\lambda,d_{h},d_{l} }} \rightarrow_{p} \hat{\rho}(\infty |d_{h},h,\lambda) = \hat{\rho}(\infty |d_{h},h) $,
    which also implies that $\frac{N_{h,d_{h},d_{l},\hat{L}_{1}}}{N_{h,d_{h},d_{l} }}$ converges in probability to $\hat{\rho}(\infty |d_{h},h)$. Since $\frac{N_{h,d_{l},d_{h}} }{N_{h,d}}$ converges in probability to $Binom(d_{h};d,q_{h}  )$, we know that
    \begin{equation} \label{eq: hatl_eq1}
        \frac{N_{h,d,\hat{L}_{1}}(g^{n}) }{N_{h,d}(g^{n})} \rightarrow_{p} \sum_{d_{h}} Binom(d_{h};d,q_{h}  ) \hat{\rho}(\infty |d_{h},h)
    \end{equation}
    
    This concludes the proof for type $h$ nodes and we now consider type $l$ nodes. Let $N_{h,\lambda',\hat{L}_{1}}(g^{n})$ denote nodes with type $h$, connectedness $\lambda'$ and are on $\hat{L}_{1}$, and let $N_{h,\lambda'}(g^{n})$ denote nodes with type $h$, connectedness $\lambda'$.
    
    The probability that a node of type $l$ and connectedness $\lambda$ is connected to a node of type $h$ and connectedness $\lambda'$ is $\frac{\lambda \lambda' q }{n}$. 
    Given types and connectedness of a pair of nodes, whether they are connected are independent of the types of all other nodes, and formation of all other links. Therefore, conditional on the type $h$ subnetwork, a node of type $l$, connectedness $\lambda$ and $d_{l}$ type $l$ neighbours is connected to $0$ nodes on $\hat{L}_{1}$ and $k$ type $h$ nodes not on $\hat{L}_{1}$ with probability
    \begin{equation} \label{eq:hatl_eq2}
          \Pi_{\lambda'}Binom(0; N_{h,\lambda',\hat{L}_{1}} , \frac{\lambda \lambda' q }{n})
    \sum_{\sum_{\lambda'} k_{\lambda'}=k } \Pi_{\lambda'} Binom(k_{\lambda'}; N_{h,\lambda'} - N_{h,\lambda',\hat{L}_{1}} ,  \frac{q\lambda \lambda'}{n})  
    \end{equation}
    Since $\frac{N_{h,\lambda',\hat{L}_{1}}}{n}$ and $\frac{N_{h,\lambda'}}{n}$ converges in probability to $\gamma_{h}f_{h}(\lambda')\hat{\rho}(\infty|h,\lambda') $ and $\gamma_{h}f_{h}(\lambda') $ respectively, using properties of Poisson distribution, we know that  \ref{eq:hatl_eq2} converges to 
    \[
   PO(k; \sum_{\lambda'} \lambda \lambda' q \gamma_{h}f_{h}(\lambda')) Binom(k;k,\frac{\sum_{\lambda'} \lambda' f_{h}(\lambda')(1-\hat{\rho}(\infty|h,\lambda'))}{\sum_{\lambda'}  \lambda'  f_{h}(\lambda')})
    \]
    Using this, and the fact that $N_{l,\lambda,d_{l}}$ goes to infinity as $n$ goes to infinity, we have
    \[
    \frac{N_{l,d_{l},d_{h},\lambda } - N_{l,d_{l},d_{h},\lambda,\hat{L}_{1}}}{N_{l,d_{l},\lambda}} \rightarrow_{p} PO(d_{h}; \sum_{\lambda'} \lambda \lambda' q \gamma_{h}f_{h}(\lambda')) Binom(d_{h};d_{h},\frac{\sum_{\lambda'} \lambda' f_{h}(\lambda')(1-\hat{\rho}(\infty|h,\lambda'))}{\sum_{\lambda'}  \lambda'  f_{h}(\lambda')})
    \]
Similarly, conditional on the type $h$ subnetwork, a node of type $l$, connectedness $\lambda$ and $d_{l}$ type $l$ neighbours is connected to $k$ type $h$ nodes with probability
    \[
    \sum_{\sum_{\lambda'} k_{\lambda'}=k } \Pi_{\lambda'} Binom(k_{\lambda'}; N_{h,\lambda'} , \frac{q\lambda \lambda'}{n} )  
    \]
    which converges to $PO(k; \sum_{\lambda'} \lambda \lambda' q \gamma_{h}f_{h}(\lambda')) $. 
    
    Therefore, $\frac{N_{l,d_{l},d_{h},\lambda}}{N_{l,d_{l},\lambda}}$ converges in probability to $PO(d_{h}; \sum_{\lambda'} \lambda \lambda' q \gamma_{h}f_{h}(\lambda'))$, and
    \[
    \frac{N_{l,d_{l},d_{h},\lambda,\hat{L}_{1}} }{N_{l,d_{l},d_{h},\lambda}} \rightarrow_{p} 1-(\frac{\sum_{\lambda'} \lambda' f_{h}(\lambda')(1-\hat{\rho}(\infty|h,\lambda'))}{\sum_{\lambda'}  \lambda'  f_{h}(\lambda')} )^{d_{h}} = \hat{\rho}(\infty|d_{h},h)
    \]
    where the last equality follows from \Cref{rem:what_is_rho} and constuction of the network formatio process associated with $\hat{\rho}(\cdot)$. This implies that $\frac{N_{l,d_{l},d_{h},\hat{L}_{1}} }{N_{l,d_{l},d_{h}}} = \hat{\rho}(\infty|d_{h},h)$, and 
    \[
    \frac{ N_{l,d,\hat{L}_{1}} }{N_{l,d} } \rightarrow_{p} \sum_{d_{h}} Binom(d_{h};d, 1-q_{l})\hat{\rho}(\infty|d_{h},h)
    \]
    The fact that both these limits are increasing follows since $\rho(\infty|d_{h},h)$ is increasing in $d_{h}$, and binomial distribution with a larger number of trials FOSD one with a smaller number of trials, given the same success probability. 
	\end{proof}

\linespread{0.8}
\small
\bibliographystyle{abbrvnat}
\addcontentsline{toc}{section}{References}
\bibliography{main}

\newpage
\linespread{1}
\normalsize

\section{Appendix for Online Publication Only}

		\subsection{Existence of limit}
I show that for $n$ large enough, sender's payoff can be arbitrarily close to the limsup of $V_{net}^{*}(n)$. To do that I first prove a useful lemma, 
\begin{lemma} \label{lemmanobadsignal}
	For any $\beta'>0$, there exists some $n_{1}$ such that for all games with $n \ge n_{1}$, if some sender's strategy achieves some payoff $V$ for the sender, then there exists another strategy $\Pi_{2}^{n}$ that achieves at least $V-\beta'$ and
	satisfies \begin{enumerate}
		\item $\Pi_{2}^{n}$ uses no signal that convinces neither type of receivers. 
		\item For all $s \in S_{good}$, it must be the case that $\pi_{3}(s|\omega=1)b(\omega=1)\mu_{l}(\omega=1) = \pi_{3}(s|\omega=0)b(\omega=0)\mu_{l}(\omega=0)$, where signals in $S_{good}$ persuades both groups. 
	\end{enumerate} 
\end{lemma}
\begin{proof}
	For any $\beta'>0$, we can find a $\hat{\beta }>0$ such that $sup( v(x+8\hat{\beta})-v(x)) < \beta'$. For this $\hat{\beta}$, we can find some $\overline{d}$ such that $\sum_{d>\overline{d}} (\gamma_{h}p_{d}^{h}+\gamma_{l}p_{d}^{l}	) < \hat{\beta} $. Fix this $\overline{d}$. We can find some $\beta < \hat{\beta}$ such that $4\beta + \frac{4\beta \zeta(t,d)}{\zeta(l,1)} \le \zeta(t,d) - \hat{\zeta}(t,d)$ for all $t$ and $d \le \overline{d}$, and $\frac{\mu_{l}(\omega=1)}{\mu_{l}(\omega=0)}\frac{1- 2 \beta \frac{\mu_{h}(\omega=1)}{\mu_{h}(\omega=0)}}{1- 2 \beta } < 1$.
	By lemmas proved in the main text, we can find some $n_{1}$ such that for all $n>n_{1}$, the following conditions are satisfied with probability at least $1-\beta$:
	\begin{enumerate}
		\item $\sum_{d > \overline{d}}  \frac{N_{h,d}+N_{l,d}}{n}  < \hat{ \beta }  $
		\item $\frac{N_{t,d}}{n}  \in [\gamma_{t}p_{d}(1-\beta) ,  \gamma_{t}p_{d} (1+ \beta)]  $ for $t\in \{h,l\}$ and all $d \le \overline{d}$. 
		\item $g^{n} \in \mathcal{A}_{n}(\beta,\overline{d} )  $
	\end{enumerate}
	
	Let $\Pi_{1}^{n}$ achieve payoff $V$. If the conditions stated in lemma are satisfied by $\Pi_{1}^{n}$. We are done. Suppose this is not the case.
	
	In the game with $n$ receivers, sender's payoff from $\Pi^{n}_{1}$ is
	
	\begin{multline} \label{sec6eq1}
		\sum_{\omega} \mu_{s}(\omega) \sum_{s \in S \cup \{\emptyset\}} \pi_{1}(s|\omega)\bigg[\sum_{g^{n}}Pr(g^{n}) v(\sum_{t,d}\frac{N_{t,d,s}(g^{n})}{n}\alpha(s,t,d;\Pi^{n}_{1}) +  (1-\frac{N_{t,d,s}(g^{n})}{n})\alpha(\emptyset,t,d;\Pi^{n}_{1})  )    
		\bigg]
	\end{multline}
	By how $n$ is chosen, \ref{sec6eq1} is smaller than
	\begin{multline} \label{sec6eq5}
		\sum_{\omega}\mu_{s}(\omega) \big[ (1-\sum_{s \in S_{good}} \pi_{1}(s|\omega)\mathbf{1}(Z(s) \cup L_{1} = \emptyset )  - 	\sum_{s \in S_{int}}\pi_{1}(s|\omega)\mathbf{1}(Z(s) \cup \hat{L}_{1} = \emptyset )v\big(\sum_{t,d \le \overline{d} }\gamma_{t}p_{d} \alpha(\emptyset,t,d;\Pi^{n}_{1})\big) +\\
		\sum_{s \in S_{good}} \pi_{1}(s|\omega)\mathbf{1}(Z(s) \cup L_{1} \ne \emptyset )
		v\big(  \sum_{t,d \le \overline{d} }p_{d}\gamma_{t}
		\zeta(t,d) +  p_{d}\gamma_{t}(1-\zeta(t,d))\alpha(\emptyset,t,d;\Pi^{n}_{1}) \big) + \\
		\sum_{s \in S_{int}}\pi_{1}(s|\omega)\mathbf{1}(Z(s) \cup \hat{L}_{1} \ne \emptyset )
		v\big(  \sum_{t,d \le \overline{d} }p_{d}\gamma_{t}
		\hat{\zeta}(t,d) +  p_{d}\gamma_{t}(1-\hat{\zeta}(t,d))\alpha(\emptyset,t,d;\Pi^{n}_{1}) \big) 
		\big]  +\beta  v(1) +\beta'
	\end{multline}
	Going through essentially the same steps, we can also show that it's bounded below by
	\begin{multline} \label{sec6eq4}
			\sum_{\omega}\mu_{s}(\omega) \big[ (1-\sum_{s \in S_{good}} \pi_{1}(s|\omega)\mathbf{1}(Z(s) \cup L_{1} = \emptyset )  - 	\sum_{s \in S_{int}}\pi_{1}(s|\omega)\mathbf{1}(Z(s) \cup \hat{L}_{1} = \emptyset )v\big(\sum_{t,d \le \overline{d} }\gamma_{t}p_{d} \alpha(\emptyset,t,d;\Pi^{n}_{1})\big) +\\
		\sum_{s \in S_{good}} \pi_{1}(s|\omega)\mathbf{1}(Z(s) \cup L_{1} \ne \emptyset )
		v\big(  \sum_{t,d \le \overline{d} }p_{d}\gamma_{t}
		\zeta(t,d) +  p_{d}\gamma_{t}(1-\zeta(t,d))\alpha(\emptyset,t,d;\Pi^{n}_{1}) \big) + \\
		\sum_{s \in S_{int}}\pi_{1}(s|\omega)\mathbf{1}(Z(s) \cup \hat{L}_{1} \ne \emptyset )
		v\big(  \sum_{t,d \le \overline{d} }p_{d}\gamma_{t}
		\hat{\zeta}(t,d) +  p_{d}\gamma_{t}(1-\hat{\zeta}(t,d))\alpha(\emptyset,t,d;\Pi^{n}_{1}) \big) 
		\big]  -\beta  v(1) -\beta'
	\end{multline}
	
	I first show that we can find some $\Pi_{2}^{n}$ such that $\mathbf{1}(Z(s) \cup L_{1} \ne \emptyset )  $ for all $s \in S_{good}$. If this is the case for $\Pi_{1}^{n}$, we are done. Suppose not and let $S_{good}'$ denote the set of signals in $S_{good}$ such that $Z(s) \cup L_{1} \ne \emptyset$. 
	Consider sender's strategy $\Pi_{2}^{n}$ that send $\emptyset$ when $\Pi_{1}^{n}$ sends signals in $S_{good }\setminus S_{good}'$, and is otherwise the same as $\Pi_{1}^{n}$.
    
	One can check that by construction receivers' action upon each observation is the same as when they face $\Pi_{1}^{n}$. Therefore, sender's payoff is bounded below by expression $\ref{sec6eq4}$.

	Let's denote elements in $S_{good}'$ by $\{s_{1}',s_{2}',...\}$ and denote $1-\sum_{s \in S_{good}' \cup S_{int} } \pi_{2}(s | \omega=0)$ by $K$. 
	Now, we construct $\Pi_{3}^{n}$ from $\Pi_{2}^{n}$ using the following procedure: 
	Starting from $j=1$,
	\begin{enumerate}
		\item Let $\pi_{3}(s_{j}'|\omega=0) = \frac{\pi_{2}(s_{j}'|\omega=1) \mu_{l}(\omega=1) }{\mu_{l}(\omega=0)  }$ if 
		\begin{multline}
			\frac{\pi_{2}(s_{j}'|\omega=1) \mu_{l}(\omega=1) }{\mu_{l}(\omega=0)  } - \pi_{2}(s_{j}'|\omega=0) \\ \le K -\sum_{i=1}^{j-1}  ( \pi_{3}(s_{i}'|\omega=0) - \pi_{2}(s_{i}'|\omega=0)   )
		\end{multline}
		\item Otherwise, set 
		\[
		\pi_{3}(s_{j}'|\omega=0)	=	\pi_{2}(s_{j}'|\omega=0) + K -\sum_{i=1}^{j-1}  ( \pi_{3}(s_{i}'|\omega=0) - \pi_{2}(s_{i}'|\omega=0))
		\]
	\end{enumerate}
	Set $\pi_{3}(s|\omega=1) = \pi_{2}(s|\omega=1)$ for all $s \in S_{good}'$.
	Set $\pi_{3}(s|\omega) = \pi_{2}(s|\omega)$ for all $s \in S_{int}$ and send $\emptyset$ with the complementary probability. Again, let seeding strategies remain the same. 
	
	One can easily verify that by construction receivers choose action 1 upon observing signals in $S_{good}'$ and type $h$ receivers choose action 1 upon observing signals in $S_{int}$. Also, since we are not changing the probabilities in state $\omega=1$, the probability that receivers observe $\emptyset$ in state $\omega=1$ is unchanged. In state $\omega=0$, we increased the probabilities that signals in $S_{good}'$ are realised and visible by $K$, and this decreases the probability of observing $\emptyset$ in state $\omega$ by at least $K \beta$. The probabilities of signals in $S_{int}$ being realised and visible are unchanged in state $\omega=0$. We reduced the probability of signals in $S_{bad}$ being realised and visible with probability by at most $K$, and this increases the probability of observing $\emptyset$ by at most $\beta K$. Therefore, compared to $\Pi_{2}^{n}$, receivers must observe $\emptyset$ in state 0 with weakly lower probability and thus must have weakly higher priors upon observing $\emptyset$. We also know that 
	\[
	v\big(  \sum_{t,d \le \overline{d} }p_{d}\gamma_{t}
	\zeta(t,d) +  p_{d}\gamma_{t}(1-\zeta(t,d))\alpha(\emptyset,t,d;\Pi^{n}_{1}) \big) \ge v\big(\sum_{t,d \le \overline{d} }\gamma_{t}p_{d} \alpha(\emptyset,t,d;\Pi^{n}_{1})\big)
	\]
	Therefore, sender's payoff from $\Pi_{3}^{n}$ is bounded below by expression $\ref{sec6eq4}$.

$\Pi_{3}^{n}$ will surely satisfy condition 1 in the lemma. If it also satisfies condition 2 in the lemma, the proof is finished. Suppose not, then by how $\Pi_{3}^{n}$ is constructed, we must have $\pi_{3}(\emptyset|\omega=0)=0$.  

	Conditional on $s \in S_{good}$ being realised, the probability that a receiver with type $t$ and degree $d \le \overline{d}$ observes it with probability at least
	\[
	\zeta(t,d) -2\beta 
	\]
	Similarly, conditional on some $s \in S_{int}$ being realised, the probability that a receiver with type $t$ and degree $d \le \overline{d}$ observes it with probability at most
	\[
	\hat{\zeta}(t,d) + 2 \beta
	\]
	Construct $\Pi_{5}^{n}$ from $\Pi_{3}^{n}$ in the following way. First, for all $s \in S_{good}'$, let
	\[
	\pi_{5}(s|\omega=1) =\pi_{3}(s|\omega=1)   \qquad  \pi_{5}(s|\omega=0) =\frac{\pi_{5}(s|\omega=1) \mu_{l}(\omega=1)}{\mu_{l}(\omega=0)}
	\]
	Let $K_{2} =  \sum_{s \in S_{good}'} \pi_{5}(s|\omega=0) - \pi_{3}(s|\omega=0)$. 
	Denote the elements in $S_{int}$ by $(s_{1}^{int},s_{2}^{int},...)$. Then, use the following procedure, starting from $i=1$
	\begin{enumerate}
		\item Let $\pi_{5}(s_{i}^{int}|\omega)$ be 0 if $K_{2}-\sum_{j=1}^{i} \pi_{3}(s_{j}^{int}|\omega=0) \ge 0$
		\item If $K_{2}-\sum_{j=1}^{i} \pi_{3}(s_{j}^{int}|\omega=0) < 0$, if the following ratio
		\[
	\frac{\mu_{l}(\omega=1)}{\mu_{l}(\omega=0)}	\frac{\pi_{3}(s_{i}^{int}|\omega=1)}{\pi_{3}(s_{i}^{int}|\omega=0) - (K-\sum_{j=1}^{i-1} \pi_{3}(s_{j}^{int}|\omega=0))}	
		\]
		is larger than 1, then set $\pi_{5}(s_{i}^{int}|\omega)=0$. Otherwise, set $\pi_{5}(s_{i}^{int}|\omega=1)=\pi_{3}(s_{i}^{int}|\omega=1)$ and $\pi_{5}(s_{i}^{int}|\omega=0)$ to be 
		\[
		\pi_{3}(s_{i}^{int}|\omega=0) - (K-\sum_{j=1}^{i-1} \pi_{3}(s_{j}^{int}|\omega=0))
		\]
		\item If $K-\sum_{j=1}^{i} \pi_{3}(s_{j}^{int}|\omega=0) > 0$ and $i < |S_{int}|$, continue steps above for $i+1$. Otherwise, terminate the process. Let $i_{end}$ be the index where the process is terminated.
	\end{enumerate}
	 For all $i > i_{end}$, let $\pi_{5}(s_{i}^{int}|\omega)=\pi_{3}(s_{i}^{int}|\omega)$. Send $\emptyset$ in state $\omega=1$ with the complementary probability. Let seeding strategy be unchanged from before. 
	 Again, by construction, all receivers choose action 1 upon observing signals in $S_{good}'$ and type $h$ receivers choose action 1 upon observing signals in $S_{int}$.
	
	I now argue that type $l$ receivers' action must be weakly higher under $\Pi_{5}$ compared to $\Pi_{3}$. Compared to $\Pi_{3}^{n}$, the probabilities of observing $\emptyset$ changed due to the following changes
	\begin{enumerate}
		\item In state 0, probability of observing signals in $S_{good}'$ increased by $K_{2}$, and probability of observing signals in $S_{int}$ decreased by $K_{2}$. Due to these changes, the probability of observing $\emptyset$ in state 0 is decreased by $(\zeta(t,d) -4 \beta -\hat{\zeta}(t,d))K_{2} >0$ for receiver with degree $d$ and type $t$.
		\item Probability of observing $\emptyset$ in state 1 increases due to decreases in probabilities of observing signals $s_{i}^{int}$ with $i < i_{end}$
		\item We know that likelihood ratio upon observing $\emptyset$ under $\Pi_{5}^{n}$ must be higher than that under $\Pi_{3}^{n}$ after these changes.
		\item Finally, we may need to decrease the probability of observing $s_{i_{end}}^{int}$. If it is decreased, then we know that the probability of observing it in state 1 is decreased by $\pi_{3}(s_{i_{end}}^{int})E(N_{s_{i_{end}}^{int},t,d}/N_{t,d} |s_{i_{end}}^{int}) $, and that in state 0 is decreased by some probability smaller than $\frac{\pi_{3}(s_{i_{end}}^{int}) \mu_{l}(\omega=1)}{\mu_{l}(\omega=0)}E(N_{s_{i_{end}}^{int},t,d}/N_{t,d} |s_{i_{end}}^{int}) $. Therefore, after being multiplied by type $l$ receiver's prior for their respective states, probability of observing $\emptyset$ decreases more in state 1 than in state 0. This implies that if likelihood ratio of observing $\emptyset$ is higher than 1 before the change, it must also be higher than 1 after the change. 
	\end{enumerate}
	Also, since $\Pi_{3}^{n}$ sends no $\emptyset$ in state 0. and all signals used must convince type $h$ receivers, one can verify that all receivers with type $h$ must choose action 1 upon observing $\emptyset$. Therefore, payoff from $\Pi_{5}^{n}$ must also be weakly higher than that from $\Pi_{3}^{n}$, and thus bounded be low by \ref{sec6eq4} again. 
	The proof is concluded since $\beta$ and $\beta'$ can be arbitrarily small and $v(1)$ is bounded. 
\end{proof}
I can show that under the following condition, $\lim_{n \rightarrow \infty} V_{net}^{*}(D^{(n)})$ exists. \\
\textbf{Condition A1}:
There exists no $d$ such that 
\[
\frac{\mu_{l}(\omega=1)}{ \mu_{h}(\omega=0) }\frac{ 1-\hat{\zeta}(l,d)\frac{\mu_{h}(\omega=0) }{\mu_{h}(\omega=1)}   }{ 1-\hat{\zeta}(l,d)  } =1
\]

From now on, suppose that condition A1 holds. 
For any  $\beta'>0$, we can find a $\hat{\beta}>0$ such that $sup( v(x+8\hat{\beta})-v(x)) < \beta'$. We can also find a $\beta < \hat{\beta}$, such that 
\[
\frac{\mu_{l}(\omega=1)}{ \mu_{h}(\omega=0) }\frac{ 1-(\hat{\zeta}(l,d) - 2\beta)\frac{\mu_{h}(\omega=0) }{\mu_{h}(\omega=1)}   }{ 1-(\hat{\zeta}(l,d) - 2\beta)}- 1
\]  
has the same sign as\footnote{This is possible because of condition A1.}
\[
\frac{\mu_{l}(\omega=1)}{ \mu_{h}(\omega=0) }\frac{ 1-(\hat{\zeta}(l,d) + 2\beta)\frac{\mu_{h}(\omega=0) }{\mu_{h}(\omega=1)}   }{ 1-(\hat{\zeta}(l,d) +2\beta)}- 1
\] and $\beta$ also satisfies 
\[
\frac{\mu_{l}(\omega=1)}{\mu_{l}(\omega=0)}\frac{ 1- \frac{\mu_{h}(\omega=0) }{ \mu_{h}(\omega=1)} 2 \beta}{1-   + 2 \beta } < 1
\]
For this $\beta$, we can find some $\overline{d}$ such that $\sum_{d>\overline{d}} (\gamma_{h}p_{d}^{h}+\gamma_{l}p_{d}^{l}	) <\beta $. Fix this $\overline{d}$. By lemmas above and construction, we can find some $n_{1}$ such that for all $n>n_{1}$, the following conditions are satisfied with probability $1-\beta$:
\begin{enumerate}
	\item $\sum_{d > \overline{d}}  \frac{N_{h,d}+N_{l,d}}{n}  < \beta  $
	\item $\frac{N_{t,d}}{n}  \in [\gamma_{t}p_{d}(1-\beta) ,  \gamma_{t}p_{d} (1+ \beta)]  $ for $t\in \{h,l\}$ and all $d \le \overline{d}$. 
	\item $g^{n} \in \mathcal{A}_{n}(\beta,\overline{d} )  $
	\item If there's some strategy $\Pi_{1}^{n}$ that achieves payoff $V$, then there exists another strategy $\Pi_{2}^{n}$ that achieves at least $V-\beta$ and satisfies conditions stated in lemma \ref{lemmanobadsignal}.
\end{enumerate}

By definition and point 4 above, we know that we can find some $n \ge n_{1}$ and some strategy $\Pi_{1}^{n}$ that achieves payoff at least $limsup_{n} V^{*}_{net}(D^{n}) -2\beta $, and it satisfies the conditions in lemma \ref{lemmanobadsignal}. Let $S_{good}$ denote the set of signals that convince both types of receivers, $S_{int}$ denote the set of signals that convince only type $h$ receivers. 

Following the same steps in proof of lemma \ref{lemmanobadsignal}, sender's payoff from $\Pi^{n}_{1}$ is bounded above by \ref{sec6eq5}
Similarly, given any strategy $\Pi^{n}$ that satisfies conditions in \ref{lemmanobadsignal}, sender's payoff is bounded below by \ref{sec6eq4}.

We consider two cases. First, suppose that no type $l$ receivers with degree less than $\overline{d}$ chooses action 1 upon observing $\emptyset$ with strategy $\Pi_{1}^{n}$. We can find some $d_{max}$ such that type $h$ receiver with degree $d_{max}$ chooses 1 upon observing $\emptyset$ and those with degrees $d_{max}< d \le \overline{d}$ choose 0 upon observing $\emptyset$. We know that the probability that receivers with type $h$ and degree $d_{max}$ observe signals in $S_{good}$ with probability at least
\begin{equation}
 max\{0,\zeta(h,d_{max})\sum_{s \in S_{good}}\pi_{2}(s|\omega=1)\mathbf{1}(Z(s) \cup L_{1}\ne \emptyset)-2\beta\} \equiv OB_{good}
\end{equation}
 in state $\omega=1$.
Similarly, they observe signals in $S_{int}$ with probability at least in state $\omega=1$
\[
max\{0,\hat{\zeta}(h,d_{max})\sum_{s \in S_{int}}\pi_{2}(s|\omega=1)\mathbf{1}(Z(s) \cup \hat{L}_{1}\ne \emptyset)- 2\beta	\} \equiv OB_{int}
\]
Therefore, the likelihood ratio upon observing $\emptyset$ of receiver with type $h$ and degree $d_{max}$ mast satisfy
\[
\frac{\mu_{h}(\omega=1)}{\mu_{h}(\omega=0)} \frac{1-OB_{good} - OB_{int}   }{ 1-\frac{\mu_{l}(\omega=1)}{\mu_{l}(\omega=0)}OB_{good}  - min\{1-\frac{\mu_{l}(\omega=1)}{\mu_{l}(\omega=0)}OB_{good},\frac{\mu_{h}(\omega=1)}{\mu_{h}(\omega=0)}OB_{int}\}  } \ge 1
\]
For any game with $n \ge n_{2}$ players, consider the following strategy for sender
\[
\pi_{2}(s_{good}|\omega=1) = OB_{good}/(\zeta(h,d_{max})+ 2 \beta ) 
\]
\[
\pi_{2}(s_{int}|\omega=1) = OB_{int}/(\hat{\zeta}(h,d_{max})+2 \beta )
\]
\[
\pi_{2}(s_{good}|\omega=0) =  \frac{\pi_{2}(s_{good}|\omega=1) \mu_{l}(\omega=1)}{\mu_{l}(\omega=0)} 
\]
\[
\pi_{2}(s_{int}|\omega=0) =  min \{\frac{\pi_{2}(s_{int}|\omega=1) \mu_{h}(\omega=1)}{\mu_{h}(\omega=0)} , 1-\pi_{2}(s_{good}|\omega=0) \}
\]
Send $\emptyset$ with the complementary probability in each state.
Let the seed always hit relevant giant component. 
If $\pi_{2}(\emptyset|\omega=0)=0$, then one can easily verify that all receivers with type $h$ must choose action 1 upon observing $\emptyset$. Otherwise, receiver with degree $d_{max}$ and type $h$ observes $s_{good}$ in state 1 with probability at most $OB_{good}$
and $s_{int}$ in state 1 with probability at most $OB_{int}$.
Therefore, likelihood ratio upon observing $\emptyset$ is at least
\[
\frac{\mu_{h}(\omega=1)}{\mu_{h}(\omega=0)} \frac{1-OB_{good} - OB_{int}   }{ 1-\frac{\mu_{l}(\omega=1)}{\mu_{l}(\omega=0)}OB_{good}  - \frac{\mu_{h}(\omega=1)}{\mu_{h}(\omega=0)}OB_{int}  } \ge 1
\]
Therefore, we know that receivers with both types choose action 1 upon observing $s_{good}$, and type $h$ receivers choose 1 upon observing $s_{int}$. Type $h$ receivers with degree weakly less than $d_{max}$ choose action 1 upon observing $\emptyset$. We can bound sender's payoff below by
\begin{multline} \label{sec6eq6}
	-\beta v(1) - \beta'+ \sum_{\omega}\mu_{s}(\omega) \big[ (1-\pi_{2}(s_{good}|\omega)  - \pi_{2}(s_{int}|\omega))v\big(\sum_{t,d \le \overline{d} }\gamma_{t}p_{d} \alpha(\emptyset,t,d;\Pi^{n}_{1})\big) +\\
	\pi_{2}(s_{good}|\omega)
	v\big(  \sum_{t,d \le \overline{d} }p_{d}\gamma_{t}
	\zeta(t,d) +  p_{d}\gamma_{t}(1-\zeta(t,d))\alpha(\emptyset,t,d;\Pi^{n}_{1}) \big) + \\
	\pi_{2}(s_{int}|\omega)
	v\big(  \sum_{t,d \le \overline{d} }p_{d}\gamma_{t}
	\hat{\zeta}(t,d) +  p_{d}\gamma_{t}(1-\hat{\zeta}(t,d))\alpha(\emptyset,t,d;\Pi^{n}_{1}) \big) 
	\big]  
\end{multline}
From how the strategies are constructed, we know that 
\[
\sum_{s \in S_{good}}\pi_{1}(s|\omega=1)\mathbf{1}(Z(s) \cup L_{1}\ne \emptyset)	- \pi_{2}(s_{good}|\omega=1)   \le 4\beta/(\zeta(h,d=1))
\]
We also know that 
\[
\sum_{s \in S_{good}}\pi_{1}(s|\omega=0)\mathbf{1}(Z(s) \cup L_{1}\ne \emptyset)	- \pi_{2}(s_{good}|\omega=0)   \le \frac{\mu_{l}(\omega=1) }{\mu_{l}(\omega=0) }4\beta/(\zeta(h,d=1))
\]
Similarly, we have
\[
\sum_{s \in S_{int}}\pi_{1}(s|\omega=1)\mathbf{1}(Z(s) \cup \hat{L}_{1}\ne \emptyset)	- \pi_{2}(s_{int}|\omega=1)   \le min\{1,4\beta/(\hat{\zeta}(h,d=1))\}
\]
and
\[
\sum_{s \in S_{int}}\pi_{1}(s|\omega=0)\mathbf{1}(Z(s) \cup \hat{L}_{1}\ne \emptyset)	- )\pi_{2}(s_{int}|\omega=0)   \le	min\{1,\frac{\mu_{h}(\omega=1) }{\mu_{h}(\omega=0) } 4\beta/(\hat{\zeta}(h,d=1))\}
\]
Also, because we know that $ v\big(\sum_{t,d \le \overline{d} }\gamma_{t}p_{d} \alpha(\emptyset,t,d;\Pi^{n}_{1})\big) \le v\big(  \sum_{t,d \le \overline{d} }p_{d}\gamma_{t}
\hat{\zeta}(t,d) +  p_{d}\gamma_{t}(1-\hat{\zeta}(t,d))\alpha(\emptyset,t,d;\Pi^{n}_{1}) \big)  \le v\big(  \sum_{t,d \le \overline{d} }p_{d}\gamma_{t}
\zeta(t,d) +  p_{d}\gamma_{t}(1-\zeta(t,d))\alpha(\emptyset,t,d;\Pi^{n}_{1}) \big)$, we know that sender's payoff from $\Pi_{2}$ is at least
\begin{multline} 
	\sum_{\omega}\mu_{s}(\omega) \big[ (1-\sum_{s \in S_{good}} \pi_{1}(s|\omega)\mathbf{1}(Z(s) \cup L_{1}\ne \emptyset)  - 	\sum_{s \in S_{int}}\pi_{1}(s|\omega)\mathbf{1}(Z(s) \cup \hat{L}_{1}\ne \emptyset))v\big(\sum_{t,d \le \overline{d} }\gamma_{t}p_{d} \alpha(\emptyset,t,d;\Pi^{n}_{1})\big) +\\
	 \sum_{s \in S_{good}} \pi_{1}(s|\omega)\mathbf{1}(Z(s) \cup L_{1}\ne \emptyset) 
	v\big(  \sum_{t,d \le \overline{d} }p_{d}\gamma_{t}
	\zeta(t,d) +  p_{d}\gamma_{t}(1-\zeta(t,d))\alpha(\emptyset,t,d;\Pi^{n}_{1}) \big) + \\
	\sum_{s \in S_{int}}\pi_{1}(s|\omega)\mathbf{1}(Z(s) \cup \hat{L}_{1}\ne \emptyset)
	v\big(  \sum_{t,d \le \overline{d} }p_{d}\gamma_{t}
	\hat{\zeta}(t,d) +  p_{d}\gamma_{t}(1-\hat{\zeta}(t,d))\alpha(\emptyset,t,d;\Pi^{n}_{1}) \big) 
	\big]  \\
	- \bigg(4\beta/(\zeta(h,d=1)) + \frac{\mu_{l}(\omega=1) }{\mu_{l}(\omega=0) }4\beta/(\zeta(h,d=1)) \bigg)v(1)-\beta v(1)-\beta'  \\
-min\{2,4\beta/(\hat{\zeta}(h,d=1))\bigg(1	 + \frac{\mu_{h}(\omega=1) }{\mu_{h}(\omega=0) } \bigg)\} \\ (  v\big(  \sum_{t,d \le \overline{d} }p_{d}\gamma_{t}
\hat{\zeta}(t,d) +  p_{d}\gamma_{t}(1-\hat{\zeta}(t,d))\alpha(\emptyset,t,d;\Pi^{n}_{1}) \big) -v\big(\sum_{t,d \le \overline{d} }\gamma_{t}p_{d} \alpha(\emptyset,t,d;\Pi^{n}_{1})\big) )
\end{multline}
If $\hat{\zeta}(h,d=1) > 0$, then the result follows since $\beta$ and $\beta'$ can be arbitrarily small and $v(.)$ is bounded. If $\hat{\zeta}(h,d=1) = 0$, then we must have $\hat{\zeta}(h,d) = 0$ for all $d$, and therefore the last term of the expression above is equal to 0. The result then follows since again $\beta$ and $\beta'$ can be arbitrarily small and $v(.)$ is bounded.

Now, we consider the second case, where there exists some $d_{min} \le \overline{d}$ such that all type $l$ receivers with degree lower than $d_{min}$ choose action 0 upon observing $\emptyset$, and those with $d_{min}$ choose action 1 upon observing $\emptyset$.
Receivers with type $t$ and degree $d_{min}$ observe signals in $S_{int}$ in state 0 with probability at most
\[
\hat{\zeta}(l,d_{min}) \sum_{s \in S_{int}} \pi_{1}(s|\omega=0)\mathbf{1}(Z(s) \cup \hat{L}_{1}\ne \emptyset) + 2\beta
\]
This implies that their likelihood ratio conditional on not observing signals in $S_{int}$ (i.e. conditional on observing signals in $S_{good}$ or $\emptyset$) is at most
\[
\frac{\mu_{l}(\omega=1)}{\mu_{l}(\omega=0)}\frac{ 1- \frac{\mu_{h}(\omega=0) }{ \mu_{h}(\omega=1)}( \hat{\zeta}(l,d_{min}) \sum_{s \in S_{int}} \pi_{1}(s|\omega=0)\mathbf{1}(Z(s) \cup \hat{L}_{1}\ne \emptyset) + 2 \beta)}{1-  (\hat{\zeta}(l,d_{min})\sum_{s \in S_{int}} \pi_{1}(s|\omega=0)\mathbf{1}(Z(s) \cup \hat{L}_{1}\ne \emptyset) + 2 \beta) } 
\]
which must be weakly greater than 1 since these receivers choose action 1 upon observing signals in $S_{good}$ or $\emptyset$. Also, this implies that $\frac{|\hat{L}_{1}|}{n}$ must not converge to 0, and thus $\hat{ \zeta }(t,d)$ must be positive 0 for all $t$ and $d$. This implies that 
\begin{equation} \label{equation76}
	\frac{\mu_{l}(\omega=1)}{ \mu_{h}(\omega=0) }\frac{ 1-(\hat{\zeta}(l,d_{min}) - 2\beta)\frac{\mu_{h}(\omega=0) }{\mu_{h}(\omega=1)}   }{ 1-(\hat{\zeta}(l,d_{min}) - 2\beta)} \ge 1
\end{equation}

Consider the following strategy $\Pi_{3}^{n}$ for any $n \ge n_{1}$. 
\[\pi_{3}(s_{int} |\omega=0) = min\{1, \frac{\hat{\zeta}(l,d_{min})}{\hat{\zeta}(l,d_{min})- 2\beta}\sum_{s \in S_{int}} \pi_{1}(s|\omega=0)\mathbf{1}(Z(s) \cup \hat{L}_{1}\ne \emptyset) + \frac{2\beta}{ \hat{\zeta}(l,d_{min})- 2\beta}\}\]
\[
\pi_{3}(s_{good}|\omega =0) = 1 -\pi_{3}(s_{int}|\omega =0)
\]
\[
\pi_{3}(s_{good}|\omega =1)=\frac{\mu_{l}(\omega=0)}{\mu_{l}(\omega=1)} \pi_{3}(s_{good} =0)
\]
\[
\pi_{3}(s_{int}|\omega =1)= \frac{\mu_{h}(\omega=0)}{\mu_{h}(\omega=1)}  \pi_{3}(s_{good} =0)
\]
Send $\emptyset$ with the complementary probabilities in each state. 

If $\pi_{3}(s_{int} |\omega=0)=1$, then all receivers with degrees higher than $d_{max}$ observe $s_{int}$ in state 0 with probability at least $\hat{\zeta}(l,d_{min})-2\beta$ and thus will choose
will choose action 1 upon observing $\emptyset$ by equation \ref{equation76}. If $\pi_{3}(s_{int} |\omega=0)b_{3}(\omega=0) \ne 1$,
receivers with type $l$ and degree $d_{min}$ or higher observe $s_{int}$ in state $\omega$ with probability at least
\[
\hat{\zeta}(l,d_{min})sum_{s \in S_{int}} \pi_{1}(s|\omega=0)\mathbf{1}(Z(s) \cup \hat{L}_{1}\ne \emptyset) + 2\beta
\]
Therefore, the likelihood ratio upon not observing $s_{int}$ (i.e. observing $s_{good}$ or $\emptyset$) is at least 
\[
\frac{\mu_{l}(\omega=1)}{\mu_{l}(\omega=0)}\frac{ 1- \frac{\mu_{h}(\omega=0) }{ \mu_{h}(\omega=1)}(\hat{\zeta}(l,d_{min})\sum_{s \in S_{int}} \pi_{1}(s|\omega=0)\mathbf{1}(Z(s) \cup \hat{L}_{1}\ne \emptyset) + 2\beta)}{1-  (\hat{\zeta}(l,d_{min})\sum_{s \in S_{int}} \pi_{1}(s|\omega=0)\mathbf{1}(Z(s) \cup \hat{L}_{1}\ne \emptyset) + 2\beta)} 
\]
Then, since the posterior upon observing $s_{good}$ is exactly 0.5, we know that the posterior upon observing $\emptyset$ must also be higher than 0.5. 

Therefore, sender's payoff from $\Pi_{3}^{n}$ is bounded below by
\begin{multline} 
	-\beta v(1)-\beta'+ \sum_{\omega}\mu_{s}(\omega) \big[ (1-\pi_{3}(s_{good}|\omega)  - 	\pi_{3}(s_{int}|\omega)v\big(\sum_{t,d \le \overline{d} }\gamma_{t}p_{d} \alpha(\emptyset,t,d;\Pi^{n}_{1})\big) +\\
	\pi_{3}(s_{good}|\omega)
	v\big(  \sum_{t,d \le \overline{d} }p_{d}\gamma_{t}
	\zeta(t,d) +  p_{d}\gamma_{t}(1-\zeta(t,d))\alpha(\emptyset,t,d;\Pi^{n}_{1}) \big) + \\
	\pi_{3}(s_{int}|\omega)
	v\big(  \sum_{t,d \le \overline{d} }p_{d}\gamma_{t}
	\hat{\zeta}(t,d) +  p_{d}\gamma_{t}(1-\hat{\zeta}(t,d))\alpha(\emptyset,t,d;\Pi^{n}_{1}) \big) 
	\big]  
\end{multline}
By construction of $\Pi_{3}^{n}$,
\begin{multline}
	\sum_{\omega}\mu_{s}(\omega) \pi_{3}(s_{int}|\omega)  -\sum_{\omega}\mu_{s}(\omega)\sum_{s \in S_{int}}\pi_{1}(s|\omega)\mathbf{1}(Z(s) \cup \hat{L}_{1}\ne \emptyset)  \le \\
	(\mu_{s}(\omega=0)+\mu_{s}(\omega=1)\frac{\mu_{h}(\omega=1) }{\mu_{h}(\omega=0) }  )\frac{4\beta}{(\hat{\zeta}(l,d=1)-2\beta)}   
\end{multline}
Since $\sum_{s \in S_{good}}\pi_{1}(s|\omega)\mathbf{1}(Z(s) \cup L_{1}\ne \emptyset) \le 1-\sum_{s \in S_{int}}\pi_{1}(s|\omega) \mathbf{1}(Z(s) \cup \hat{L}_{1}\ne \emptyset)$ , We also know that 
\begin{multline}
	\sum_{\omega}\mu_{s}(\omega)\sum_{s \in S_{good}}\pi_{1}(s|\omega)\mathbf{1}(Z(s) \cup L_{1}\ne \emptyset) -  \sum_{\omega}\mu_{s}(\omega) \pi_{3}(s_{good}|\omega)  \le \\
	(\mu_{s}(\omega=0)+\mu_{s}(\omega=1)\frac{\mu_{l}(\omega=1) }{\mu_{l}(\omega=0) }  )\frac{4\beta}{(\hat{\zeta}(l,d=1)-2\beta)}  
\end{multline}

Therefore, sender's payoff from $\Pi_{3}^{n}$ is bounded below by
\begin{multline} 
	\sum_{\omega}\mu_{s}(\omega) \big[ (1- \sum_{s \in S_{good}} \pi_{1}(s|\omega)\mathbf{1}(Z(s) \cup L_{1}\ne \emptyset)   - \sum_{s \in S_{int}}\pi_{1}(s|\omega)\mathbf{1}(Z(s) \cup \hat{L}_{1}\ne \emptyset) )v\big(\sum_{t,d \le \overline{d} }\gamma_{t}p_{d} \alpha(\emptyset,t,d;\Pi^{n}_{1})\big) +\\
	 \sum_{s \in S_{good}} \pi_{1}(s|\omega)\mathbf{1}(Z(s) \cup L_{1}\ne \emptyset)  
	v\big(  \sum_{t,d \le \overline{d} }p_{d}\gamma_{t}
	\zeta(t,d) +  p_{d}\gamma_{t}(1-\zeta(t,d))\alpha(\emptyset,t,d;\Pi^{n}_{1}) \big) + \\
	\sum_{s \in S_{int}}\pi_{1}(s|\omega)\mathbf{1}(Z(s) \cup \hat{L}_{1}\ne \emptyset) 
	v\big(  \sum_{t,d \le \overline{d} }p_{d}\gamma_{t}
	\hat{\zeta}(t,d) +  p_{d}\gamma_{t}(1-\hat{\zeta}(t,d))\alpha(\emptyset,t,d;\Pi^{n}_{1}) \big) 
	\big] \\- \beta v(1)-\beta' -\frac{4\beta}{(\hat{\zeta}(l,d=1)-2\beta)}v(1)(2\mu_{s}(\omega=0)+\mu_{s}(\omega=1)\frac{\mu_{l}(\omega=1) }{\mu_{l}(\omega=0) }   + \mu_{s}(\omega=1)\frac{\mu_{h}(\omega=1) }{\mu_{h}(\omega=0) }  )
\end{multline}

Recall that receiver's payoff from $\Pi_{1}$ is bounded above by
\begin{multline} 
	\sum_{\omega}\mu_{s}(\omega) \big[ (1- \sum_{s \in S_{good}} \pi_{1}(s|\omega)\mathbf{1}(Z(s) \cup L_{1}\ne \emptyset)   - 	\sum_{s \in S_{int}}\pi_{1}(s|\omega)\mathbf{1}(Z(s) \cup \hat{L}_{1}\ne \emptyset) )v\big(\sum_{t,d \le \overline{d} }\gamma_{t}p_{d} \alpha(\emptyset,t,d;\Pi^{n}_{1})\big) +\\
	 \sum_{s \in S_{good}} \pi_{1}(s|\omega)\mathbf{1}(Z(s) \cup L_{1}\ne \emptyset)  
	v\big(  \sum_{t,d \le \overline{d} }p_{d}\gamma_{t}
	\zeta(t,d) +  p_{d}\gamma_{t}(1-\zeta(t,d))\alpha(\emptyset,t,d;\Pi^{n}_{1}) \big) + \\
	\sum_{s \in S_{int}}\pi_{1}(s|\omega)\mathbf{1}(Z(s) \cup \hat{L}_{1}\ne \emptyset) 
	v\big(  \sum_{t,d \le \overline{d} }p_{d}\gamma_{t}
	\hat{\zeta}(t,d) +  p_{d}\gamma_{t}(1-\hat{\zeta}(t,d))\alpha(\emptyset,t,d;\Pi^{n}_{1}) \big) 
	\big]  + \beta v(1) + \beta'
\end{multline}
Since $v(1)$ is bounded, and we can choose $\beta$ and $\beta'$ to be arbitrarily small, The result follows. 

Let $\mathcal{P}$ denote the set of network signal strategies that:
\begin{enumerate}
    \item use only two non-empty signals $s$ and $s'$ with positive probability.
    \item probabilities satisfy $\pi(s|\omega=1)\mu_{l}(\omega=1)=\pi(s|\omega=0)\mu_{l}(\omega=0)$ and $\pi(s'|\omega=0) =min\{\frac{\pi(s'|\omega=1)\mu_{l}(\omega=1)}{\mu_{l}(\omega=0)},1-\pi(s|\omega=0)\} $.
    \item seed of $s$ is always located on the giant component and seed of $s'$ is always on the giant component in the subnetwork of type $h$ nodes. 
\end{enumerate}

The proof above shows that there exists an $\overline{n}$ large enough such that for any $n_{1}>\overline{n}$, the supremum payoff for that $n_{1}$ can be approximately achieved for all $n>\overline{n}$. This shows that sender can achieve the limsup of $V_{net}^{*}(n)$, which by definition implies that limit exists for $V_{net}^{*}(n)$. 

In addition, the result above implies that when $n$ gets large, we can get arbitrarily close to $V_{net}^{n} = \lim_{n \rightarrow \infty}V_{net}^{*}(n)$ by using some strategy in the set $\mathcal{P}$ for all $n$. Consider any $\Pi \in \mathcal{P}$. By \Cref{lemma:non_empty_posterior}, all receivers choose action 1 upon observing $s$, and only type $h$ receivers choose action 1 upon observing $s'$. This implies that a receiver with type $t$ and degree $d$ has a likelihood ratio for $\emptyset$ that converges to 
\begin{equation} \label{eq:online_app_likelihood}
    \frac{\mu_{t}(\omega=1)( 1 -  \pi(s|\omega=1)\zeta(t,d) - \pi(s'|\omega=1)\hat{\zeta}(t,d)  )    }{ \mu_{t}(\omega=0)( 1 -  \pi(s|\omega=0)\zeta(t,d) - \pi(s'|\omega=0)\hat{\zeta}(t,d)  )}
\end{equation}

If this ratio converges to below (above) 1, then action of a receiver with type $t$ and degree $d$ converges to 0 (1). I now argue that if the ratio is exactly 0.5, then it's WLOG to assume that receiver's action converges to 1 when evaluating payoffs. First, suppose that the type with limiting likelihood ratio equal to 1 is type $h$. This implies that $\pi(s|\omega) \ne 0$. The sender could then reduce $\pi(s|\omega)$ by an arbitrarily small proportion, which by definition of $\pi(s|\omega)$ in $\mathcal{P}$ will increase the likelihood ratio. 

Now suppose that the type is $l$. If $\pi(s'|\omega=1)\mu_{h}(\omega=1)>\pi(s'|\omega=0)\mu_{h}(\omega=0)$, then sender can reduce $\pi(s'|\omega=1)$ by an arbitrarily small amount. This raises the lieklihood ratio while still keeping $\pi(s'|\omega=1)\mu_{h}(\omega=1)>\pi(s'|\omega=0)\mu_{h}(\omega=0)$. If $\pi(s'|\omega=1)\mu_{h}(\omega=1)= \pi(s'|\omega=0)\mu_{h}(\omega=0)$, the sender can increase $\pi(s'|\omega)$ by some arbitrarily small proportion. If in addition $\pi(\emptyset|\omega)>0$ in both states, the increase in $\pi(s'|\omega)$ can be made small enough to make sure that $\pi(s|\omega) + \pi(s'|\omega) \le 1$. If $\pi(\emptyset|\omega)=0$ in some state, sender can reduce $\pi(s|\omega)$ by an arbitrarily small proportion to make sure that $\pi(s|\omega) + \pi(s'|\omega) \le 1$ in both states. In either case, the change described raises the likelihood ratio. 

Let $\Pi'$ denote the modified strategy. By how the changes are constructed, limiting posteriors upon observing $\emptyset$ under $\Pi'$ are not equal to 0.5, and are larger than 0.5 iff they are weakly larger than 0.5 under $\Pi$. Therefore, under $\Pi'$ and when $n$ gets large, receivers choose action 1 upon observing $\emptyset$ iff their limiting posterior upon observing $\emptyset$ is weakly larger than 0.5 under $\Pi$. Since the change in $\pi$ is arbitrarily small, sender's limiting payoff from $\Pi'$ will be arbitrarily close to her payoff from $\Pi$, assuming agents' actions upon observing $\emptyset$ converge to 1 when the limiting likelihood ratio is 1. 

Therefore, sender's limiting supermum payoff is solution to the following maximisation problem. Also, when $n$ gets large, sender can use a strategy in $\mathcal{P}$ that's arbitrarily close to the maximiser to get arbitrarily close to the supremum payoff.
\begin{multline} \label{eq:online_app_maxprob}
   \max_{\Pi \in \mathcal{P}} \sum_{\omega}\mu_{s}(\omega)\big[\pi(s|\omega)v(\sum_{d,t}p_{d}^{t}(\zeta(t,d)   +(1-\zeta(t,d))a^{*}(\emptyset;t,d) )  ) +  \\ \pi(s'|\omega)v(\sum_{d}p_{d}^{h}\hat{\zeta}(h,d)   +  \sum_{d,t}p_{d}^{t} (1-\hat{\zeta}(t,d))a^{*}(\emptyset;t,d) )  + (1-  \pi(s'|\omega)-\pi(s|\omega))v( \sum_{d,t}p_{d}^{t} a^{*}(\emptyset;t,d) )        \big] 
\end{multline}
subject to the condition that $a^{*}(\emptyset;t,d)=1$ if and only if the ratio $\ref{eq:online_app_likelihood}$ is weakly larger than 1.

I now discuss what could happen when condition A1 does not hold. Suppose that condition A1 is violated, and optimal solution to the problem above is given by $\pi(s|\omega)=0$ and $\pi(s'|\omega=0)=1$ and $\pi(s'|\omega=1) = \frac{\mu_{h}(\omega=0)}{\mu_{h}(\omega=1)}$. It's possible that this limiting payoff could not be achieved without putting seeds outside the giant component, but can be achieved with seeds outside the giant component. 

Since condition A1 is violated, there exists a degree $d$ such that 
\begin{equation} \label{eq: online_app_eqd21}
    \frac{\mu_{l}(\omega=1)( 1  - \frac{\mu_{h}(\omega=0)}{\mu_{h}(\omega=1)}\hat{\zeta}(t,d)  )    }{ \mu_{l}(\omega=0)( 1 - \hat{\zeta}(t,d)  )} = 1
\end{equation}
In any game with $n$ receivers, under the signal structure with $\pi(s|\omega)=0$, $\pi(s'|\omega=0)=1$ and $\pi(s'|\omega=1) = \frac{\mu_{h}(\omega=0)}{\mu_{h}(\omega=1)}$, receivers with type $t$ and degree $d$ have a posterior upon observing $\emptyset$
\[
 \frac{\mu_{l}(\omega=1)E( 1  - \frac{\mu_{h}(\omega=0)}{\mu_{h}(\omega=1)}\frac{N_{t,d,\hat{L}_{1}} }{N_{t,d} }  )    }{ \mu_{l}(\omega=0)E( 1 -  \frac{N_{t,d,\hat{L}_{1}} }{N_{t,d} }  )}
\]
which will converge to expression $\ref{eq: online_app_eqd21}$. It's possible that $\frac{N_{t,d,\hat{L}}}{N_{t,d}}$ converges to $\hat{\zeta}(t,d)$ from below and never reaches it, or fluctuates around it. In either case, there will exist infinitely many values of $n$ above any $\overline{n}$ where using strategies in $\mathcal{P}$ will not achieve solution from the reduced form maximisation problem \ref{eq:online_app_maxprob}. 

However, by using strategies not in $\mathcal{P}$, sender may be able to persuade agents of type $l$ and degree $d$ to take action 1 upon observing $\emptyset$. Let $N_{l,d,max}$ denote the number of nodes outside the giant component that observe $s'$, when seeds are chosen optimally to maximise the number of nodes observing it. In the game with $n$ receivers, upon observing $\emptyset$, receivers with type $l$ and degree $d$ have a likelihood ratio
\begin{equation} \label{eq:online_app_eq22}
     \frac{\mu_{l}(\omega=1)E( 1  - \frac{\mu_{h}(\omega=0)}{\mu_{h}(\omega=1)}(\frac{N_{l,d,\hat{L}_{1}} }{N_{l,d} } +\frac{N_{l,d,max}}{n})  )    }{ \mu_{l}(\omega=0)E( 1 -  \frac{N_{l,d,\hat{L}_{1}} }{N_{l,d} }  -\frac{N_{l,d,max}}{N_{l,d}})}
\end{equation}
which can be larger than 1. 

When $\frac{N_{l,d,max}}{n}$ is large enough (i.e. expression \ref{eq:online_app_eq22} is larger than 1), then the sender can achieve a payoff arbitrarily close to result from the reduced form maximisation problem \ref{eq:online_app_maxprob}. However, as discussed, this cannot be achieved with strategies in $\mathcal{P}$, and the sender must choose seeds outside the giant component appropriately. In large networks, this can be extremely difficult. When expression \ref{eq:online_app_eq22} is smaller than 1, then solution of the reduced form maximisation problem cannot be achieved. 

Finally, another point worth noting is that whether expression \ref{eq:online_app_eq22} is larger than 1 for infinitely many values of $n$ depends on the relative speeds at which the random networks converge to its limiting distributions and how fast the number of seeds as a proportion of $n$ converges to 0. More specifically, we need to look at how fast $\frac{N_{t,d,\hat{L}_{1}} }{N_{t,d} } $ converges to $\hat{\zeta}(t,d)$ and $\frac{N_{l,d,max}}{n}$ converges to 0. However, it's difficult to interpret both speeds. The sequence of networks (and games) is a mathematical construct to formally capture large social networks. The sequence itself and how the network gets large do not capture anything in real world networks. For seeds, the results in the main text only requires seeds to be vanishing as a fraction of total population, and number $n^{\alpha}$ is chosen purely for convenience. In other words, the speed of how fast seeds as a fraction of $n$ converge to 0 is completely arbitrary. Given these, the discussion on optimal strategies when condition A1 is violated should be seen as illustrating the possibility that seeds reaching a vanishing fraction of population could still be important for sender's payoff and optimal strategies. The precise conditions (on relative speeds of convergence) are not very informative. 

\end{document}